\documentclass[cleveref, autoref, thm-restate,11pt,letterpaper,onecolumn,accepted=2024-03-21]{quantumarticle}
\pdfoutput=1

\usepackage{mathtools}
\usepackage{amsmath}
\usepackage{amssymb}
\usepackage{amsthm}
\usepackage[linesnumbered,ruled,vlined]{algorithm2e}
\usepackage{upgreek}
\usepackage{tikz}
\usepackage{graphicx}
\usepackage{forest}
\usepackage{mathrsfs}
\usetikzlibrary{graphs,quotes}
\usepackage{fancyhdr}
\usepackage{varioref}
\usepackage{verbatim} 
\usepackage{multicol}
\usepackage{enumerate}
\usepackage[normalem]{ulem}
\usepackage{caption}
\usepackage{subcaption}
\usepackage{float}
\usepackage[margin=1in]{geometry}
\usepackage{academicons}
\definecolor{orcidlogocol}{HTML}{A6CE39}

\usepackage{mathrsfs}
\usepackage{mathtools}

\usepackage{url}
\usepackage{xspace}
\usepackage{thm-restate}

\usepackage[TU]{fontenc}

\newtheorem{theorem}{Theorem}
\newtheorem{lemma}[theorem]{Lemma}

\newtheorem{definition}[theorem]{Definition}

\newtheorem*{conjecture*}{Conjecture}

\theoremstyle{definition}

\usepackage{microtype}

\usepackage{tikz}
\usetikzlibrary{calc}

\usepackage{epstopdf}

\usepackage{hyperref}
\hypersetup{
 colorlinks = {true},
 pdftitle={Improved Quantum Query Complexity on Easier Inputs}
}
\usepackage{xcolor}
\definecolor{darkred}  {rgb}{0.5,0,0}
\definecolor{darkblue} {rgb}{0,0,0.5}
\definecolor{darkgreen}{rgb}{0,0.5,0}
\hypersetup{
  urlcolor   = blue,         
  linkcolor  = darkblue,     
  citecolor  = darkgreen,    
  filecolor  = darkred       
}

\usepackage[capitalise]{cleveref}
\crefname{lemma}{Lemma}{Lemmas}
\crefname{proposition}{Proposition}{Propositions}
\crefname{definition}{Definition}{Definitions}
\crefname{theorem}{Theorem}{Theorems}
\crefname{conjecture}{Conjecture}{Conjectures}
\crefname{corollary}{Corollary}{Corollaries}
\crefname{section}{Section}{Sections}
\crefname{appendix}{Appendix}{Appendices}
\crefname{figure}{Fig.}{Figs.}
\crefname{algorithm}{Alg.}{Algs.}
\crefname{equation}{Eq.}{Eqs.}
\crefname{table}{Table}{Tables}
\crefname{claim}{Claim}{Claims}
\crefname{line}{Line}{Lines}
\crefname{item}{Item}{Items}

\newcommand{\sop}[1]{{\mathcal #1}}
\newcommand{\ket}[1]{|#1\rangle}
\newcommand{\bra}[1]{\langle#1|}
\newcommand{\proj}[1]{|#1\rangle\!\langle#1|}
\newcommand{\ketbra}[2]{|{#1}\rangle\!\langle{#2}|}
\newcommand{\braket}[2]{\langle{#1}|{#2}\rangle}
\DeclareMathOperator*{\argmin}{arg\,min}
\newcommand{\Lal}[1]{{\Lambda^{\!{#1}}}}
\newcommand{\U}[3]{{U({#1},{#2},{#3})}}
\newcommand{\w}[2]{{w({#1},{#2})}} 
\renewcommand{\wp}[2]{{w_+({#1},{#2})}}
\newcommand{\wm}[2]{{w_-({#1},{#2})}}
\newcommand{\wpo}[1]{{w_+({#1})}}
\newcommand{\wmo}[1]{{w_-({#1})}}
\newcommand{\Aa}[1]{{\tilde{A}^{#1}}}
\newcommand{\PtHx}{\Pi_{x}}
\newcommand{\PHx}{\Pi_{{H}(x)}}
\newcommand{\Lalp}{{\Uplambda^{\!\alpha,\hat{\varepsilon}}}}
\newcommand{\PHxp}{\Uppi_{x}}
\newcommand{\Up}[4]{{\sop U({#1},{#2},{#3},{#4})}}
\newcommand{\iner}[1]{\langle{#1}|{#1}\rangle}

\newcommand{\PrCont}[1][i]{Pr(cont\textrm{ }#1)} 
\newcommand{\PrErr}[1][i]{Pr(err\textrm{ }#1)} 
\newcommand{\PrFinal}{Pr(final)} 
\def\polylog{\operatorname{polylog}}

\title{Improved Quantum Query Complexity on Easier Inputs}
\author[1]{Noel T. Anderson}
\author[1]{Jay-U Chung}
\thanks{Current affiliation: IBM T. J. Watson Research Center, Yorktown Heights, NY, USA}
\author[1]{Shelby Kimmel}
\orcid{0000-0003-0726-4167}
\email{skimmel@middlebury.edu}
\author[3]{Da-Yeon Koh}
\thanks{Contributed to this work during an internship at Middlebury College}
\author[1,4]{Xiaohan Ye}
\affil[1]{Middlebury College, Middlebury, VT, USA}
\affil[3]{Williams College, Williamstown, MA, USA}
\affil[4]{Brown University, Providence, RI, USA}

\date{}

\begin{document}

\maketitle

\vspace{-.5cm}
\begin{abstract} 
Quantum span program algorithms for function evaluation sometimes have reduced query complexity when promised that the input has a certain structure. We design a modified span program algorithm to show these improvements persist even without a promise ahead of time, and we extend this approach to the more general problem of state conversion.  As an application, we prove exponential and superpolynomial quantum advantages in average query complexity for several search problems, generalizing Montanaro's Quantum Search with Advice [Montanaro, TQC 2010]. 
\end{abstract}


\section{Introduction}\label{sec:intro} 

Quantum algorithms often perform
 better when given a promise on the input. For example, if we know that there
 are $M$ marked items out of $N$, or no marked items at all, then Grover's
 search can be run in time and query complexity $O(\sqrt{N/M})$, rather than
 $O(\sqrt{N})$, the worst case complexity with a single marked item \cite
 {groverQuantumMechanicsHelps1997a,aharonovQuantumComputation1999}.

In the case of Grover's algorithm, a series of results
\cite{boyerTightBoundsQuantum1998,brassardQuantumAmplitudeAmplification2000,brassardQuantumCounting1998}
removed the promise; if there are $M$ marked items, there is a quantum
search algorithm that runs in $O(\sqrt{N/M})$ complexity, even without
knowing the number of marked items ahead of time. Most relevant for our
work, several of these algorithms involve iteratively running Grover's
search with exponentially growing runtimes \cite
{boyerTightBoundsQuantum1998,brassardQuantumAmplitudeAmplification2000}
until a marked item is found.

Grover's algorithm was one of the first quantum query algorithms
discovered \cite{groverQuantumMechanicsHelps1997a}. Since that time, span
programs and the dual of the general adversary bound were developed,
providing frameworks for creating optimal query algorithms for function
decision problems
\cite{reichardtSpanProgramsQuantum2009,reichardtReflectionsQuantumQuery2011}
and nearly optimal algorithms for state conversion problems, in which the goal is to generate a quantum state based on an oracle and an input state \cite{leeQuantumQueryComplexity2011}. 
Moreover, these frameworks are also useful in practice \cite{beigiQuantumSpeedupBased2019,belovsSpanProgramsQuantum2012,belovsSpanProgramsFunctions2012,belovsTightQuantumLower2020a,cadeTimeSpaceEfficient2018,delorenzoApplicationsQuantumAlgorithm2019a}.

For some span program algorithms, analogous to multiple marked items in
Grover's search, there are features which, if promised to exist, allow for
improvement over the worst case query complexity. For example, a span program
algorithm for deciding $st$-connectivity uses $O\left(n^{3/2}\right)$ queries on an
$n$-vertex graph. However, if \textrm{promised} that the shortest path, if it exists,
has length at most $k$, then the problem can be solved with $O(\sqrt
{k}n)$ queries \cite{belovsSpanProgramsQuantum2012}.

Our contribution is to remove the requirement of the promise; we improve
the query complexity of generic span program and state conversion algorithms in the
case that some speed-up inducing property (such as multiple marked items or a 
short path) is present, even without knowing about the structure in advance.
One might expect this is
trivial: surely if an algorithm produces a correct result with fewer queries
when promised a property is present, then it should also produce a correct
result with fewer queries without the promise if the property still holds?
While this is true and these algorithms always output a result, even if run
with fewer queries, the problem is that they don't produce a flag of
completion, and their output cannot always be easily verified. Without a flag of
completion or a promise of structure, it is impossible to be confident that
the result is correct. Span program and state conversion algorithms differ
from Grover's algorithm in their lack of a flag; in Grover's algorithm one can
use a single query to test whether the output is a marked item, thus flagging
that the output of the algorithm is correct, and that the algorithm has run
for a sufficiently long time. We note that when span program algorithms
previously have claimed an improvement with structure, they always included a
promise, or they give the disclaimer that running the algorithm will be
incorrect with high probability if the promise is not known ahead of time to
be satisfied, e.g. Ref. \cite[App. C.3]{cadeTimeSpaceEfficient2018}.

We use an approach that is similar to the
iterative modifications to Grover's algorithm; we run subroutines for
exponentially increasing times, and we have novel ways to flag when the
computation should halt. On the hardest inputs, our algorithms 
match the asymptotic performance of existing bounded error algorithms. On easier inputs, our approach on average matches the asymptotic performance, up to log factors, of existing algorithms when
those existing algorithms additionally have an optimal promise.

Because our algorithms use fewer queries on easier inputs without needing to
know they are easier inputs, they provide the possibility of improved
average query complexity over input oracles when there is a distribution of easier and
harder inputs. In this direction, we generalize a result by Montanaro that
showed a super-exponential quantum advantage in average query complexity for
the problem of searching for a single marked item under a certain
distribution \cite{montanaro2010quantum}. In particular, we provide a
framework for proving similar advantages using quantum algorithms based on
classical decision trees, opening up the potential for a broader range of
applications than the approach used by Montanaro. We apply this technique to
prove an exponential and superpolynomial quantum advantage in average query
complexity for searching for multiple items and searching for the first
occurring marked items, respectively.

Where prior work showed improvements for span program algorithms with a
promise, our results immediately provide an analogous improvement without the
promise:
\begin{itemize}
\item For undirected $st$-connectivity described above, our algorithm determines whether there is a path from $s$ to
$t$ in an $n$-vertex graph with $\tilde{O}(\sqrt{k}n)$ queries if there is a path of length $k$,
 and if there is no path,
the algorithm uses $\tilde{O}(\sqrt{nc})$ queries, where $c$ is the size of
the smallest cut between $s$ and $t$. In either case, $k$ and $c$ need not be known ahead of time.
\item For an $n$-vertex undirected graph, we can determine if it is connected in $\tilde{O}(n\sqrt{R})$ queries, where $R$ is the average
effective resistance, or not connected in $\tilde{O}(\sqrt{n^{3}/\kappa})$ queries, where $\kappa$ is the number of components. These query complexities hold without knowing $R$ or $\kappa$ ahead of time. See Ref.
\cite{jarretQuantumAlgorithmsConnectivity2018} for the promise version of this problem.
\item For cycle detection on an $n$-vertex undirected graph,
whose promise version was analyzed in Ref.
\cite{delorenzoApplicationsQuantumAlgorithm2019a}, if the circuit rank is $C$, then our algorithm will detect a cycle in $\tilde{O}(\sqrt{n^{3}/C})$ queries, while if there is
no cycle and at most $\mu$ edges, the algorithm will decide there is no cycle
in $\tilde{O}(\mu\sqrt{n})$ queries. This holds
without knowing $C$ or $\mu$ ahead of time.
\end{itemize}

To achieve our results for decision problems, we modify the original span program function
evaluation algorithm to create two one-sided error subroutines. In the
original span program algorithm, the final measurement tells you with high
probability whether $f(x)=1$ or $f(x)=0.$ In one of our subroutines, the final
measurement certifies that with high probability $f(x)=1$, providing our flag
of completion, or it signals that more queries are needed to determine
whether $f(x)=1$. The other behaves similarly for $f(x)=0$. By interleaving
these two subroutines with exponentially increasing queries, we achieve our
desired performance.

The problem is more challenging for state conversion, as the standard version
of that algorithm does not involve any measurements, and so there is nothing
to naturally use as a flag of completion. We thus design a novel probing
routine that iteratively tests exponentially increasing query complexities
until a sufficient level is reached, before then running an algorithm similar
to the original state conversion algorithm.

While we analyze query complexity, the algorithms we create have average time
complexity on input $x$ that scales like $O(T_U\mathbb{E}[Q_x])$, where $\mathbb{E}[Q_x]$ is the average query
complexity on input $x$, and $T_U$ is the time complexity of implementing an
input-independent unitary. Since the existing worst-case span program and
state conversion algorithms have time complexities that scale as $O
(\max_x T_U \mathbb{E}[Q_x])$, our algorithms also improve in average time complexity
relative to the original algorithms on easier inputs. For certain problems, like
$st$-connectivity \cite{belovsSpanProgramsQuantum2012} and search 
\cite{cornelissen2020span}, it is known that $T_U=\tilde{O}(1)$, meaning that the
query complexities of our algorithms for these problems match the time
complexity up to log factors.

\subsection{Directions for Future Work}

Ambainis and de Wolf show that while there is no
quantum query advantage for the problem of \textsc{majority} in the worst case, on 
average there is a quadratic quantum advantage \cite{ambainis2001average}. However, 
their quantum algorithm uses a technique that is specific to the problem of 
\textsc{majority}, and it is not clear how it might extend to other problems. On
 the other hand, since our approach is based on span programs, a generic optimal 
 framework, it may
provide opportunities of proving similar results for more varied problems.

In the original state conversion algorithm, to achieve an error of
$\varepsilon$ in the output state (by some metric), the query complexity
scales as $O\left(\varepsilon^{-2}\right)$ \cite{leeQuantumQueryComplexity2011}. In our result, the query
complexity scales as $O\left(\varepsilon^{-5}\right)$. While this does not
matter for applications like discrete function evaluation, as considered in \cref{sec:tree_application}, in cases where
accuracy must scale with the input size, this error term could overwhelm any
advantage from our approach, and so it would be beneficial to improve this
error scaling.

Ito and Jeffery
\cite{itoApproximateSpanPrograms2019} give an algorithm to estimate the positive witness size
(a measure of how easy an instance is) with fewer queries on easier inputs. While there are similarities between
our approaches, neither result seems to directly imply the other. Better
understanding the relationship between these strategies could lead to
improved algorithms for determining properties of input structure for both
span programs and state conversion problems.

Our work can be contrasted with the work of Belovs and Yolcu 
\cite{belovs2023one}, which also has a notion of reduced query complexity on
easier inputs.  Their 
work focuses on the ``Las Vegas query complexity,'' which is
related to the amount of the state that
a controlled version of the oracle acts on over the course of the algorithm,
and which is an input-dependent quantity. They show the ``Monte Carlo query complexity,'' what we call the query complexity, is bounded by the Las Vegas query complexity of the worst-case input. 
We suspect that using techniques similar to those in our work, it would 
be possible to modify their algorithm
to obtain an algorithm with input-dependent average
query complexity 
that scales roughly with the geometric mean of the Las Vegas and Monte Carlo
complexities for that input, without
knowing anything about the input ahead of time.

\section{Preliminaries}\label{sec:Prelims}

\textbf{Basic Notation:} 
For $n>2$, let $[n]$ represent $\{1,2,\dots,n\}$, while for $n=2$, $[n]=\{0,1\}$. We use $\log$ to denote base 2 logarithm. For set builder notation like $\{r_z:z\in Z\}$ we will frequently use $\{r_z\}_{z\in Z}=\{r_z\}$, where we drop the subscript
outside the curly brackets if clear from context. We
denote a linear operator from the space $V$ to the space $U$ as $\sop L
(V,U)$. We use $I$ for the identity operator. (It will be clear from context
which space $I$ acts on.) Given a projection $\Pi$, its complement is
$\overline{\Pi}=I-\Pi.$ For a matrix $M$, by $M_{xy}$ or $(M)_{xy}$, we
denote the element in the $x^\textrm{th}$ row and $y^\textrm{th}$ column of $M$. By 
$\tilde{O}$, we denote big-O notation that ignores log factors. The $l_2$-norm of a
vector $\ket{v}$ is denoted by $\|\ket{v}\|$. For any unitary $U$, let
$P_\Theta(U)$ be the projection onto the eigenvectors of $U$ with phase at
most $\Theta$. That is, $P_\Theta(U)$ is the projection onto 
$\textrm{span}\{\ket{u}:U\ket{u}=e^{i\theta}\ket{u}\textrm{ with }|\theta|\leq\Theta\}$. For a function $f:D\rightarrow [m]$, we define $f^{-1}(b)=\{x\in D:f(x)=b\}.$

\subsection{Quantum Algorithmic Building Blocks}

We consider quantum query algorithms, in which one can access
a unitary $O_x$, called the oracle, which encodes a string $x\in X$ for
$X\subseteq [q]^n$, $q\geq 2$. The oracle acts on the
Hilbert space $\mathbb{C}^n\otimes \mathbb{C}^q$ as $O_x\ket{i}\ket{b}=\ket
{i}\ket{x_i+b\textrm{ mod }q}$, where $x_i\in[q]$ is the $i^\textrm{th}$ element of
$x$.

Given $O_x$ for $x\in X$,
we would like to perform a computation that depends on $x$. The query
complexity is the minimum number of uses of the oracle required such that for
all $x\in X$, the computation is successful with some desired probability of
success. We denote by $\mathbb{E}[Q_x]$ the average number of queries used by the algorithm on input $x$ where the expectation is over the algorithm's internal randomness. 
Given a probability distribution $\{p_x\}_{x\in X}$ over the elements of $X$, then 
$\sum_{x\in X}p_x\mathbb{E}[Q_x]$ is the average quantum query complexity of performing the 
computation with respect to $\{p_x\}.$

Several of our key algorithmic subroutines use a parallelized
version of phase estimation
\cite{magniezSearchQuantumWalk2011}, in which for a unitary $U$, a precision $\Theta>0$,
and an accuracy $\epsilon>0$, a circuit $D(U)$ implements
$O(\log\frac{1}{\epsilon})$ copies of the phase estimation circuit on $U$,
each to precision $O(\Theta)$, that all measure the phase of a single state on
the same input register. If $U$ acts on a Hilbert Space $\sop H$, then $D(U)$
acts on the space $\sop H_A\otimes((\mathbb{C}^{2})^{\otimes b})_B$ for
$b=O\left(\log\frac{1}{\Theta}\log\frac{1}{\epsilon}\right)$, where we have  used
$A$ to label the register that stores the input state, and $B$ to label the
registers that store the results of the parallel phase estimations. 

The circuit $D(U)$ can be used for Phase Checking: applying $D(U)$ 
to $\ket{\psi}_A\ket{0}_B$ and then measuring register $B$ 
in the standard basis; the probability of outcome $\ket{0}_B$ provides
information on whether $\ket{\psi}$ is close to an eigenvector of $U$ that has eigenphase close to
$0$ (in particular, with eigenphase within $\Theta$ of $0$). To characterize this
probability, we define $\Pi_0(U)$ to be the orthogonal projection onto the
subspace of $\sop H_A\otimes((\mathbb{C}^{2})^{\otimes b})_B$ that $D(U)$ maps
to states with $\ket{0}_B$ in the $B$ register.  That is,
$\Pi_0(U)=D(U)^\dagger \left(I_A\otimes\proj{0}_B\right) D(U).$ (Since
$\Pi_0(U)$ depends on the choice of $\Theta$ and $\epsilon$ used in $D(U)$,
those values must be specified, if not clear from context, when discussing
$\Pi_0(U)$.)  We now summarize  prior results for Phase Checking
 in \cref{lem:phase_det}:

\begin{lemma}[Phase Checking \cite{kitaevQuantumMeasurementsAbelian1995,cleveQuantumAlgorithmsRevisited1998,magniezSearchQuantumWalk2011}]
Let $U$ be a unitary on a Hilbert Space $\sop H$, and let $\Theta,\epsilon>0$. We call $\Theta$ the precision and $\epsilon$ the accuracy. Then there is a circuit $D(U)$ that acts on the space $\sop H_A\otimes
((\mathbb{C}^{2})^{\otimes b})_B$ for $b=O\left(\log\frac{1}{\Theta}\log\frac{1}{\epsilon}\right)$, and that
uses $O\left(\frac{1}{\Theta}\log\frac{1}{\epsilon}\right)$ calls to
control-$U$.  Then for any state $\ket{\psi}\in \sop H$
\begin{itemize}
\item $ \|P_0(U)\ket{\psi}\|^2\leq \|\Pi_0(U)\left(\ket{\psi}_A\ket{0}_B\right)\|^2\leq \|P_\Theta(U)\ket{\psi}\|^2+\epsilon,$ and
\item $\|\Pi_0(U)\left(\overline{P}_\Theta(U)\ket{\psi}\right)_A\ket{0}_B\|^2\leq \epsilon.$
\end{itemize}
\label{lem:phase_det}
\end{lemma}

We also consider implementing $D(U)$ as
described above, applying a $-1$ phase to the $A$ register if the $B$
register is \textit{not} in the state $\ket{0}_B$, and then implementing $D
(U)^\dagger$. We call this circuit Phase Reflection\footnote{In Ref.
\cite{leeQuantumQueryComplexity2011}, this procedure is referred to as ``Phase
Detection,'' but since no measurement is made, and rather only a reflection is
applied, we thought renaming this protocol as ``Phase Reflection''
would be more descriptive and easier to distinguish from ``Phase Checking.'' We apologize for any confusion this may cause when
comparing to prior work.}
and denote it as $R(U).$ Note that $R(U)=\Pi_0(U)-\overline{\Pi}_0(U)$, where $R(U)$ and 
$\Pi_0(U)$ have the same implicit precision $\Theta$ and accuracy $\epsilon$. The following 
lemma summarizes prior results on relevant properties of Phase Reflection.
\begin{lemma} [Phase Reflection \cite{magniezSearchQuantumWalk2011,leeQuantumQueryComplexity2011}]
Let $U$ be a unitary on a Hilbert Space $\sop H$, and let
$\Theta,\epsilon>0$. We call $\Theta$ the precision and $\epsilon$ the accuracy. Then there is a circuit $R(U)$ that acts on the space
$\sop H_A\otimes ((\mathbb{C}^{2})^{\otimes b})_B$ for
$b=O\left(\log\frac{1}{\Theta}\log\frac{1}{\epsilon}\right)$, and that uses
$O\left(\frac{1}{\Theta}\log\frac{1}{\epsilon}\right)$ calls to control-$U$
and control-$U^\dagger$, such that for any state $\ket{\psi}\in \sop H$
\begin{itemize}
\item $R(U)(P_0(U)\ket{\psi})\ket{0}_B=(P_0(U)\ket{\psi})_A\ket{0}_B$, and
\item $\|(R(U)+I)(\overline{P}_\Theta(U)\ket{\psi})_A\ket{0}_B\|<\epsilon $.
\end{itemize}
\label{lem:phase_refl}
\end{lemma}

We will use Iterative Quantum Amplitude Estimation \cite{Grinko2021}, which is like standard quantum amplitude estimation \cite{brassardQuantumAmplitudeAmplification2000}, but with exponentially
better success probability:

\begin{lemma} [Iterative Quantum Amplitude Estimation \cite{Grinko2021}] Let
$\delta>0$ and $\mathcal{A}$ be a quantum circuit such that on a state $\ket{\psi}$, 
$\mathcal{A}\ket{\psi}=\alpha_0\ket{0}\ket{\psi_0}+\alpha_1\ket{1}\ket{\psi_1}$. 
Then there is an algorithm that estimates $|\alpha_0|^2$
 to additive error $\delta$ with success probability at least $1-p$ using 
$O\left(\frac{1}{\delta}\log\left(\frac{1}{p}\log\frac{1}{\delta}\right)\right)$calls to $\mathcal{A}$ and $\mathcal{A}^\dagger$.
\label{lem:ampEst}
\end{lemma}

A key lemma in span program and state conversion algorithms is the effective spectral gap lemma:

\begin{lemma}[Effective spectral gap lemma, \cite{leeQuantumQueryComplexity2011}]
Let $\Pi$ and $\Lambda$ be projections, and let $U=(2\Pi-I)(2\Lambda-I)$ be the
unitary that is the product of their associated reflections. If
$\Lambda\ket{w}=0$, then $\|P_\Theta(U) \Pi\ket{w}\|\leq \frac{\Theta}{2}\|\ket{w}\|.$
\label{spec_gap_lemm}
\end{lemma}


\subsection{Span Programs}\label{sec:spanIntro}
Span programs are a tool for designing quantum query algorithms for decision problems.

\begin{definition} [Span Program] A span program is a tuple $\sop P=(H,V,\tau,A)$ on $[q]^n$ where
\begin{enumerate}
\item $H$ is a direct sum of finite-dimensional inner product spaces: $H=H_1\oplus H_2\cdots H_n\oplus H_\textrm{true}\oplus H_\textrm{false},$
and for $j\in [n]$ and $a\in[q]$, we have $H_{j,a}\subseteq H_j$, such that $\sum_{a\in [q]}H_{j,a}=H_j$.
\item $V$ is a vector space
\item $\tau\in V$ is a target vector, and
\item $A\in\sop L(H,V)$.
\end{enumerate}
Given a string $x\in[q]^n$, we use $H(x)$ to denote the subspace $H_{1,x_1}\oplus\cdots \oplus H_{n,x_n}\oplus H_\textrm{true}$, and we denote by $\PHx$ the orthogonal projection onto the space $H(x)$.

\label{def:SP}
\end{definition}

We use \cref{def:SP} for span programs because it applies to both binary and
non-binary inputs ($q\geq 2$). The definitions in Refs. \cite{belovsSpanProgramsQuantum2012,cadeTimeSpaceEfficient2018} only apply to
binary inputs ($q=2$).

\begin{definition} [Positive and Negative Witness] Given a span program $\sop P=(H,V,\tau,A)$ on
 $[q]^n$ and $x\in[q]^n$, then $\ket{w}\in H(x)$ is a positive witness for
 $x$ in $\sop P$ if $A\ket{w}=\tau$. If a positive witness exists for $x$, we
 define the positive witness size of $x$ in $\sop P$ as 
\begin{equation}
\wp{\sop P}{x}=\wpo{x}\coloneqq\min\left\{\|\ket{w}\|^2:\ket{w}\in H(x) \textrm{ and } A\ket{w}=\tau \right\}.
\end{equation}
Then $\ket{w}\in H(x)$ is an optimal positive witness for $x$ if $\|\ket{w}\|^2=\wp{\mathcal{P}}{x}$ and $A\ket{w}=\tau$.

We say $\omega\in \sop L(V,\mathbb{R})$ is a negative witness for $x$ in $\mathcal{P}$ if $\omega\tau=1$ and $\omega A \PHx=0$. If a negative witness exists for $x$, we define the negative witness size of $x$ in $\mathcal{P}$ as
\begin{equation}
 \wm{\mathcal{P}}{x}=\wmo{x}\coloneqq\min\left\{\|\omega A\|^2:
\omega\in L(V,\mathbb{R}), \omega A \PHx=0, \textrm{ and } \omega\tau=1 \right\}.
\end{equation}
Then $\omega$ is an optimal negative witness for $x$ if $\|\omega A\|^2=\wm{\mathcal{P}}{x}$, $\omega A \PHx=0,$ and $\omega\tau=1$.
\label{def:negWit}
\end{definition}
\noindent Each $x\in[q]^n$ has a positive or negative witness (but not both).

We say that a span program $\sop P$ decides the function
$f:X\subseteq[q]^n\rightarrow \{0,1\}$ if each $x\in f^{-1}(1)$ has a positive
witness in $\sop P$, and each $x\in f^{-1}(0)$ has a negative witness in $\sop P$. Then we denote the maximum positive and negative witness of $\sop P$ on $f$ as
\begin{align}
W_+(\sop P,f)=W_+\coloneqq \max_{x\in f^{-1}(1)}\wp{\sop P}{x},\qquad
W_-(\sop P,f)=W_-\coloneqq \max_{x\in f^{-1}(0)}\wm{\sop P}{x}.
\end{align}

Given a span program that decides a function, one
can use it to design an algorithm that evaluates that function with query complexity that depends on $W_+(\sop P,f)$ and $W_-(\sop P,f)$:
\begin{theorem} [\cite{reichardtSpanProgramsQuantum2009,itoApproximateSpanPrograms2019}]
For $X\subseteq[q]^n$ and $f:X\rightarrow \{0,1\}$, let $\sop P$ be a span program that decides $f$. Then there is a quantum algorithm that for any $x\in X$, evaluates $f(x)$ with bounded error, and uses $O\left(\sqrt{W_+(\sop P,f)W_-(\sop P,f)}\right)$ queries to the oracle $O_x$.
\label{thm:spEval}
\end{theorem}

Not only can any span program that decides a function $f$ be used to create a
quantum query algorithm that decides $f$, but there is always a span program
that creates an algorithm with asymptotically optimal query complexity
\cite{reichardtSpanProgramsQuantum2009,reichardtReflectionsQuantumQuery2011}.
Thus when designing quantum query algorithms for function decision problems,
it is sufficient to consider only span programs.

Given a function $f:X\rightarrow\{0,1\}$, we denote the negation of the $f$ as $f^\neg$, where $\forall x\in X, f^{\neg}(x)=\neg f(x)$.
We use a transformation that takes a span program
$\sop P$ that decides a function $f:X\rightarrow\{0,1\}$ and creates a span program $\sop P^\dagger$ that decides
$f^{\neg}$,
while preserving witness sizes for each input $x$. While such a transformation is known for Boolean span programs \cite{reichardtSpanProgramsQuantum2009}, in \cref{lem:SP_duals} we show it exists for the span programs of \cref{def:SP}. The proof is in \cref{app:sec2}.

\begin{restatable}{lemma}{SPduals}\label{lem:SP_duals}
Given a span program $\sop P=(H,V,\tau,A)$ on $[q]^n$ that decides a function
$f:X\rightarrow \{0,1\}$ for $X\subseteq [q]^n$, there is a span program
$\sop P^\dagger=(H',V',\tau',A')$ that decides $f^{\neg }$ such that  $\forall x\in f^{-1}(1), \wp{\sop P}{x}=
\wm{\sop P^\dagger}{x}$ and $\forall x\in f^{-1}(0),\wm{\sop P}{x}=
\wp{\sop P^\dagger}{x}$.
\end{restatable}


\subsection{State Conversion}\label{sec:dualAdvIntro}

In the state conversion problem, for $X\subseteq [q]^n$, we are given
descriptions of sets of pure states $\{\ket{\rho_x}\}_{x\in X}$ and 
$\{\ket{\sigma_x}\}_{x\in X}$. Then given access to an oracle for
$x$, and the quantum state $\ket{\rho_x}$, the goal is to create a state
$\ket{\sigma_x'}$ such that 
$\|\ket{\sigma_x'}-\ket{\sigma_x}\ket{0}\|\leq \varepsilon$. We call $\varepsilon$ the error of the state
conversion procedure.

Let $\rho$ and $\sigma$ be the Gram matrices of the sets 
$\{\ket{\rho_x}\}$ and $\{\ket{\sigma_x}\}$, respectively, so $\rho$
and $\sigma$ are matrices whose rows and columns are indexed by the elements
of $X$ such that $\rho_{xy}=\braket{\rho_x}{\rho_y},$ and 
$\sigma_{xy}=\braket{\sigma_x}{\sigma_y}.$ 

We now define the analogue of a span program for the problem of state conversion, which
 we call a \textit{converting vector set}:
\begin{definition} [Converting vector set]
Let $\mathscr{P}=\left(\{\ket{v_{xj}}\},\{\ket{u_{xj}}\}\right)_{x\in X, j\in[n]}$, where $\forall x\in X, j\in [n],\ket{v_{xj}},\ket{u_{xj}}\in\mathbb{C}^d$ for some $d\in \mathbb{N}$. Then we say $\mathscr P$
converts $\rho$ to $\sigma$
if it satisfies
\begin{align}
\forall x,y\in X,\quad (\rho-\sigma)_{xy}=\sum_{j\in[n]:x_j\neq y_j}\braket{u_{xj}}{v_{yj}}.\label{eq:filteredNorm}
\end{align}
We call such a $\mathscr P$ a \em{converting vector set} from $\rho$ to $\sigma$.
\label{def:converting_vector_set}
\end{definition}

Then the query complexity of state conversion is characterized as follows:

\begin{theorem} [\cite{leeQuantumQueryComplexity2011}]
Given $X\in[q]^n$ and a converting vector set $\mathscr{P}=\left(\{\ket{v_{xj}}\},\{\ket{u_{xj}}\}\right)_{x\in X, j\in[n]}$ from $\rho$ to $\sigma$, then there is quantum algorithm that on every input $x\in X$ converts 
$\ket{\rho_x}$ to $\ket{\sigma_x}$ with error $\varepsilon$ and has query complexity
\begin{align}
O\left(\max\left\{\max_{x\in X}\sum_j\|\ket{v_{xj}}\|^2,\max_{y\in X}\sum_j\|\ket{u_{yj}}\|^2\right\}\frac{\log(1/\varepsilon)}{\varepsilon^2}\right).\label{eq:filteredNormMa}
\end{align}
\label{thm:StateConv}
\end{theorem}

Analogous to witness sizes in span
programs, we define a notion of witness sizes for converting vector sets:
\begin{definition} [Converting vector set witness sizes]
Given a converting vector set
$\mathscr{P}=\left(\{\ket{u_{xj}}\},\{\ket{v_{xj}}\}\right)$, we define the
witness sizes of $\mathscr P$  as
\begin{align}
&\wp{\mathscr P}{x}\coloneqq\sum_j\|\ket{v_{xj}}\|^2 & &\textrm{positive witness size for $x$ in $\mathscr P$}\nonumber\\
&\wm{\mathscr P}{x}\coloneqq\sum_j\|\ket{u_{xj}}\|^2 & &\textrm{negative witness size for $x$ in $\mathscr P$}\nonumber\\
&W_+(\mathscr P)=W_+\coloneqq\max_{x\in X}\wp{\mathscr P}{x}& &\textrm{maximum positive witness size of $\mathscr P$}\nonumber\\
&W_-(\mathscr P)=W_-\coloneqq\max_{x\in X}\wm{\mathscr P}{x}& &\textrm{maximum negative witness size of $\mathscr P$}
\end{align}
\end{definition}

By scaling the converting vector sets, we obtain the following two results: a rephrasing of \cref{thm:StateConv} in terms of witness sizes, and a transformation that exchanges positive and negative witness sizes. Both proofs can be found in \cref{app:sec2}.
\begin{restatable}{corollary}{scaleSC}\label{cor:stateConv}
Let $\mathscr{P}$ be a converting
vector set from $\rho$ to $\sigma$ with maximum positive and negative witness sizes $W_+$ and $W_-$. 
Then there is quantum algorithm that on every input $x\in X$ converts 
$\ket{\rho_x}$ to $\ket{\sigma_x}$ with error $\varepsilon$ and uses
$O\left(\sqrt{W_+W_-}\frac{\log(1/\varepsilon)}{\varepsilon^2}\right)$ queries to $O_x$.
\end{restatable}

\begin{restatable}{lemma}{complment}\label{lem:complement}
If $\mathscr{P}$ converts $\{\ket{\rho_x}\}_{x\in X}$ to 
$\{\ket{\sigma_x}\}_{x\in X}$ , then there is a complementary
converting vector set $\mathscr P^\dagger$ that also converts $\rho$ to $\sigma$, 
such that for all $x\in X$ and for all $j\in [n]$, we have 
$w_+(\mathscr P,x)=w_-(\mathscr P^\dagger,x)$, and 
$w_-(\mathscr P,x)=w_+(\mathscr P^\dagger,x)$; the 
complement exchanges the values of the positive and negative witness sizes.
\end{restatable}


\section{Function Decision}\label{sec:func}

Our main result for function decision (deciding if $f(x)=0$ or $f(x)=1$) is the following:

\begin{restatable}{theorem}{funcEvalTheorem}\label{thm:spEvalNew}
For $X\subseteq[q]^n$, let $\sop P$ be a span program that decides $f:X \rightarrow
\left\{0,1\right\}$.  Then there is a quantum algorithm such that for any $x\in X$ and $\delta>0$
\begin{enumerate}
    \item The algorithm returns $f(x)$ with probability $1-\delta$.
    \item On input $x$, if $f(x)=1$ the algorithm uses $O\left(\sqrt{\wpo{x}W_-}\log\left(\frac{W_+}{\wpo{x}\delta}\right)\right)$ queries on average, and if $f(x)=0$
    it uses $O\left(\sqrt{\wmo{x}W_+}\log\left(\frac{W_-}{\wmo{x}\delta}\right)\right)$  queries on average.
    \item The worst-case (not average) query complexity is $O\left(\sqrt{W_+W_-}\log(1/\delta)\right)$.
\end{enumerate}
\end{restatable}

Comparing \cref{thm:spEvalNew} to \cref{thm:spEval} (which assumes constant error $\delta$), we see that in the worst case, with an input $x$ where 
$\wpo{x}=W_+$ or $\wmo{x}=W_-$, the average and worst-case performance of our algorithm is the same 
as the standard span program algorithm. However, when we have an instance $x$ with a smaller witness size,
then our algorithm has improved average query complexity, without having to know about the witness size ahead of time. 

We can also compare the query complexity of our algorithm, which does not require
a promise,  to the original span program algorithm when that algorithm is
additionally given a promise. If the original span program algorithm is
promised that, if $f(x)=1$, then $\wpo{x}=O(\sf w)$, then the
bounded error query complexity of the original algorithm on this input would
be $O(\sqrt{{\sf w} W_-})$ by \cref{thm:spEval}. On the other hand, without
needing to know ahead of time that $\wpo{x}=O(\sf w)$, our algorithm
would use $\tilde{O}(\sqrt{\sf w W_-})$ queries on this input on average, and in
fact would do better than this if $\wpo{x}=o(\sf w)$.

A key routine in our algorithm is to apply Phase Checking to a unitary 
$\U{\sop P}{x}{\alpha}$, which we describe now. We follow notation similar to that
in \cite{belovsSpanProgramsQuantum2012}. In particular,
for a span program $\sop P=(H,V,\tau,A)$ on $[q]^n$, let $\tilde{H}=H\oplus \textrm{span}\{\ket{\hat{0}}\},$ and $\tilde{H}(x)=H(x)\oplus \textrm{span}\{\ket{\hat{0}}\},$
where $\ket{\hat{0}}$ is orthogonal to $H$ and $V$. Then
we define $\Aa{\alpha} \in \sop L(\tilde{H},V)$ as
\begin{equation}\label{eq:alphaIntro}
\Aa{\alpha}=\frac{1}{\alpha}\ketbra{\tau}{\hat{0}}+A.
\end{equation}
Let $\Lal{\alpha} \in \sop L(\tilde{H},\tilde{H})$ be the orthogonal
projection onto the kernel of $\Aa{\alpha}$, and let $\PtHx\in \sop
L(\tilde{H},\tilde{H})$ be the projection onto $\tilde{H}(x).$
Finally, let
\begin{equation}
\U{\sop P}{x}{\alpha}=(2\PtHx-I)(2\Lal{\alpha}-I).
\end{equation} 
Note that
$2\PtHx-I$ can be implemented with two applications of $O_x$ 
\cite[Lemma 3.1]{itoApproximateSpanPrograms2019}, and $2\Lal{\alpha}-I$ can be implemented
without any applications of $O_x$. Queries are only made in our algorithm when we apply $\U{\sop P}{x}{\alpha}$. To analyze the query complexity, we will track
the number of applications of $\U{\sop P}{x}{\alpha}.$

The time complexity will also scale with the number of applications of $\U{\sop P}{x}{\alpha}.$ We denote the time required to implement $\U{\sop P}{x}{\alpha}$ by $T_U$, which is an input independent quantity. Since our query complexity analysis
counts the number of applications of $\U{\sop P}{x}{\alpha}$, and 
the runtime scales with the number of applications of $\U{\sop P}{x}{\alpha}$,
to bound the average time complexity of our algorithms, simply determine $T_U$
and multiply this by the query complexity.

The following lemma gives us guarantees about the results of Phase Checking
of $\U{\sop P}{x}{\alpha}$ applied to the state $\ket{\hat{0}}$:
\begin{restatable}{lemma}{phaseEstEarly}\label{lem:phase_est_early}
Let the span program $\sop P$ decide the function $f$, and let $C\geq 2$. Then for Phase Checking with unitary $\U{\sop P}{x}{\alpha}$ on the
state $\ket{\hat{0}}_A\ket{0}_B$ with error $\epsilon$ and precision $\Theta=\sqrt{\frac{\epsilon}{\alpha^2 W_-}}$,
\begin{enumerate}
\item If $f(x)=1$, and $\alpha^2\geq C\wp{\sop P}{x}$, then the probability of measuring the $B$ register to be in the state $\ket{0}_B$ is at least $1-1/C$.
\label{part:phase_est_1}
\item If $f(x)=0$ and $\alpha^2 \geq  1/W_-(\sop P,f)$, then the probability of measuring the $B$ register in the state $\ket{0}_B$ is at most $\frac{3}{2}\epsilon.$ \label{part:phase_est_2}
\end{enumerate}
\end{restatable}

\noindent Note that if $f(x)=1$ and $\alpha^2 < C\w{\sop P}{x}$, \cref{lem:phase_est_early} makes no claims about
the output. 

To prove \cref{lem:phase_est_early}, we use techniques from the
Boolean function decision algorithm of Belovs and Reichardt \cite[Section
5.2]{belovsSpanProgramsQuantum2012} and Cade et al. \cite[Section
C.2]{cadeTimeSpaceEfficient2018} and the dual adversary algorithm of Reichardt
\cite[Algorithm 1]{reichardtReflectionsQuantumQuery2011}. Our approach differs
from these previous algorithms in the addition of a parameter that controls
the precision of our phase estimation.
This approach has
not (to the best of our knowledge)\footnote{Jeffery and Ito
\cite{itoApproximateSpanPrograms2019}  also design a function decision
algorithm for non-Boolean span programs, but it has a few differences from our
approach and from that of Refs.
\cite{belovsSpanProgramsQuantum2012,cadeTimeSpaceEfficient2018}; for example,
the initial state of Jeffery and Ito's algorithm might require significant
time to prepare,  while our initial state can be prepared in $O(1)$ time.}
been applied to the non-Boolean span program formulation of \cref{def:SP}, so while
not surprising that it works in this setting, our analysis in \cref{app:func_proof} may be
of independent interest for other applications.

We use \cref{alg:bool} to prove \cref{thm:spEvalNew}.

\vspace{.3cm}
\textbf{High Level Idea of \cref{alg:bool}} The algorithm makes use of a test that, when successful, tells us that $f(x)=1$. 
However, the test is one-sided, in that failing the test does not
mean that $f(x)=0$, but instead is inconclusive. We repeatedly run this
test for both functions $f$ and $f^{\neg }$ while increasing the queries used
at each round. If we see an inconclusive result for both $f$ and $f^{\neg }$ at an
intermediate round, we can conclude neither $f(x)=1$ nor $f(x)=0$, so we
repeat the subroutine with larger queries. Once we reach a critical round that depends on $\wpo{x}$ (if $f(x)=1$) or on $\wmo{x}$ (if $f(x)=0$), the probability of an inconclusive
result becomes unlikely from that critical round onward. 

We stop iterating when the test returns a conclusive result, or when  we have passed the critical round for all $x$. While it is
unlikely that we get an inconclusive result at the final round, we return 1 if
this happens.


\begin{algorithm}
    \DontPrintSemicolon
    \SetKwInOut{Input}{Input}
    \SetKwInOut{Output}{Output}
    \Input{Span program $\mathcal{P}$ that decides a function $f$, oracle $O_x$, and failure probability $\delta$, }
    \Output{$f(x)$ with probability $1-O(\delta)$}
    $\epsilon\gets 2/9; \quad T\gets\left\lceil\log\sqrt{3W_+(\mathcal{P},f)W_-(\mathcal{P},f)}\right\rceil$\;
    \For{$i\geq 0$ \KwTo $T$ }{
    $N_i\gets18(\lceil T+\log(3/\delta)\rceil-i+1)$\;
    \tcp{Testing if $f(x)=1$}
     $\alpha_+\gets 2^{i}/\sqrt{W_-(\mathcal{P},f)}$\;
     Repeat $N_i$ times: Phase Checking of $\U{\mathcal{P}}{x}{\alpha_+}$ on $\ket{\hat{0}}_A\ket{0}_B$ with error $\epsilon$, precision $\sqrt{\frac{\epsilon}{\alpha_+^2 W_-(\mathcal{P},f)}}$ \;\label{line:check1}
     \lIf{\em{Measure $\ket{0}$ in register $B$ at least $N_i/2$ times}}{
      \Return 1\label{line:return1}
     }
     \tcp{Testing if $f(x)=0$}
     $\alpha_-\gets 2^{i}/\sqrt{W_+(\mathcal{P},f)}$\;

     Repeat $N_i$ times: Phase Checking of $\U{\mathcal{P}^\dagger}{x}{\alpha_-}$ on $\ket{\hat{0}}_A\ket{0}_B$ with error $\epsilon$, precision $\sqrt{\frac{\epsilon}{\alpha_-^2 W_+(\mathcal{P},f)}}$ \;\label{line:check0}
     \lIf{\em{Measure $\ket{0}$ in register $B$ at least $N_i/2$ times}}{
      \Return 0 \label{line:return0}
     }
    }

    \Return 1 \tcp{if exit the \textbf{for} loop without passing a test, the algorithm makes a guess}
    \caption{}
    \label{alg:bool}
\end{algorithm}





More specifically, we use Phase Checking to perform our one-sided test; we iteratively run Phase Checking on $\U{\mathcal{P}}{x}{\alpha_+}$ in \cref{line:check1} 
to check if $f(x)=1$,
and on $\U{\mathcal{P}^\dagger}{x}{\alpha_-}$ in \cref{line:check0} to check if $f(x)=0$, increasing the parameters $\alpha_{+}$ and $\alpha_-$ by a factor of 2 at each round. At some round, which we label $i^*$, $\alpha_+^2$ becomes at least $3\wp{\mathcal{P}}{x}$ or 
$\alpha_-^2$ becomes at least $3\wm{\mathcal{P}}{x}=3\wp{P^\dagger}{x}$ (by \cref{lem:complement}), depending on whether $f(x)=1$ or $f(x)=0$, respectively. 
Using \cref{lem:phase_est_early} \cref{part:phase_est_1}, from round $i^*$ onward we have a high
probability of measuring the $B$ register to be in the state 
$\ket{0}_B$ at \cref{line:check1} if $f(x)=1$ or at \cref{line:check0} if $f(x)=0$, causing the algorithm to terminate and output the correct result. \cref{part:phase_est_2} of \cref{lem:phase_est_early}  ensures that at all rounds we have a low probability of outputting the incorrect result. We don't need to
know $i^*$ ahead of time; the behavior of the algorithm will change on its
own, giving us a smaller average query complexity for instances with smaller witness size, the easier instances. 

The number of queries used by Phase Checking increases by a factor of 2 at each
round of the \textbf{for} loop. If $N_i$, the number of repetitions of 
Phase Checking at round $i$, were a constant $N$, then using a geometric series, we 
would find that the query complexity would be asymptotically equal to the queries 
used by Phase Checking in the round at which the algorithm terminates, times $N$. 
At round $i^{*}$, the round at which termination is most likely, the query complexity of Phase Checking is 
$O\left(\sqrt{\wp{\mathcal{P}}{x}W_{-}(\mathcal{P},f)}\right)$ or 
$O\left(\sqrt{\wm{\mathcal{P}}{x}W_+(\mathcal{P},f)}\right)$
(depending on if $f(x)=1$ or $f(x)=0$) by \cref{lem:phase_det}. We show the probability of continuing
to additional rounds after $i^*$ is exponentially decreasing with each extra round, so 
we find an average query complexity of $O\left(N\sqrt{\wp{\mathcal{P}}{x}W_{-}(\mathcal{P},f)}\right)$ or 
$O\left(N\sqrt{\wm{\mathcal{P}}{x}W_+(\mathcal{P},f)}\right)$ on input $x$.
Since there can be
$T=\left\lceil\log\left(\sqrt{W_+W_-}\right)\right\rceil$ rounds in \cref
{alg:bool} the worst case, this suggests that each round should have a probability of error bounded by
$O\left(T^{-1}\right)$, which we can accomplish through repetition and majority voting, but which requires $N=\Omega(\log T)$, adding an extra log factor to our query complexity. 

To mitigate this effect, we modify the number of repetitions (given by $N_i$ 
in \cref{alg:bool}) over the course 
of the algorithm so that we have a lower
probability of error (more repetitions) at earlier rounds, and a higher probability (fewer repetitions) at later
rounds. This requires additional queries
at the earlier rounds, but since these rounds are cheaper to begin with, we
can spend some extra queries to reduce our error. As a result, instead of a log factor
that depends only on $T$, we end up with a log factor that also decreases with increasing witness size, so when $\wpo{x} =
W_+$ or $\wmo{x} = W_-$, our average query complexity is at most $O\left(\sqrt
{W_+W_-}\log(1/\delta)\right)$ without any additional log factors.

\begin{proof}[Proof of \cref{thm:spEvalNew}]

We analyze \cref{alg:bool}. 

We first prove that the total success probability is at least $1 - \delta$.
Consider the case that $f(x)=0$. Let 
$i^*=\left\lceil\log\sqrt{3\wm{\sop P}{x}W_+(\sop P,f)}\right\rceil$, which is the round at which we will show our probability of exiting the \textbf{for} loop becomes large. The total number of possible iterations is $T=\left\lceil\log\sqrt{3W_+W_-}\right\rceil$, which is at least $i^*$. Let $N = \lceil T + \log(3/\delta)\rceil + 1$, so $N_i=18(N-i)$. Let $\PrCont$ denote the
probability of continuing to the next round of the \textbf{for} loop at round
$i$, conditioned on reaching round $i$,
let $\PrErr$ be the probability of returning the wrong answer at round $i$, conditioned on reaching round $i$,
and let $\PrFinal$ be the probability of reaching the end of the 
\textbf{for} loop without terminating. (Since we return $1$ if we reach the  end of
the \textbf{for} loop without terminating, this event produces an error when
$f(x)=0$.) The total probability of error is then
\begin{align}
    \label{eq:probError1} Pr(error) &= \sum_{i=0}^{T} \left(\prod_{j=0}^{i-1}\PrCont[j]\right) \cdot \PrErr + \PrFinal .
\end{align} 
We will use the probability tree diagram in \cref{fig:f0} to
 help us analyze events and probabilities. 

Since $f(x)=0$, $\PrErr$ is the probability of returning 1, which depends on
the probability of measuring $\ket{0}_B$ in Phase Checking of 
$\U{\sop P}{x}{\alpha_+}$ in \cref{line:check1} of \cref{alg:bool}. Since $\alpha_+^2W_-\geq 1$, we can use \cref{part:phase_est_2} of 
\cref{lem:phase_est_early}, to find that there is
at most a $\frac{3}{2}\epsilon=1/3$ probability of measuring $\ket{0}_B$ at
each repetition of Phase Checking. Using Hoeffding's inequality 
\cite{Hoeff63}, the probability of measuring outcome $\ket{0}$ at least 
$N_{i}/2$ times and returning $1$ in \cref{line:return1} of \cref{alg:bool} is at most
$a_i \coloneqq e^{-N_{i}/18}$. Therefore,
\begin{equation} \label{eq:ai}
    \PrErr \leq a_i
\end{equation}
which holds for all $i$ but in particular, gives us a bound on the first left branching of \cref{fig:f0}, corresponding to outputting a $1$ when $i\geq i^*$.

When $i < i^*$, we trivially bound the probability of continuing to the next round:
\begin{equation} \label{eq: earlyCont}
    \PrCont[i] \leq 1.
\end{equation}

When $i \geq i^*$, we continue to the next round when we do not return 1 in
\cref{line:return1} of \cref{alg:bool} and then do not return 0 in \cref{line:return0}, corresponding
to the two right branchings of the diagram in \cref{fig:f0}. We upper bound
the probability of the first event (first right branch in \cref{fig:f0}) by
1. To bound the probability of the second event, consider Phase Checking of
$\U{\sop P^\dagger}{x}{\alpha_-}$ in \cref{line:check0}. Since $i \geq i^*$, 
we have $\alpha_-^2\geq 3\wm{\sop P}{x}= 3\wp{\sop P^\dagger}{x}$ by \cref{lem:SP_duals}. 
Also since $W_+(\sop P,f)=W_-(\sop P^\dagger, f)$ by \cref{lem:SP_duals}, 
we have $\alpha^2\geq 1/W_-(\sop P^\dagger,f)$. Thus, as we are performing Phase Checking 
with precision $\sqrt{\epsilon/(\alpha^2W_+(\sop P,f))}=\sqrt{\epsilon/(\alpha^2W_-(\sop P^\dagger ,f))}$,
we can use \cref{part:phase_est_1} of \cref{lem:phase_est_early} with $C=3$
to conclude that the probability of measuring $\ket{0}_B$ at a single
repetition of \cref{line:check0} is at least $2/3$. Using Hoeffding's inequality 
\cite{Hoeff63}, the probability of measuring $\ket{0}_B$ more than $N_i/2$ times,
and therefore returning 0, is at least $1-e^{-N_i/18}$. Thus the probability
of not returning 0 in \cref{line:return0} is at most
\begin{align}
 1-(1-e^{-N_i/18}) = a_i.
 \label{eq:partCont}
 \end{align} 
Therefore when $i \geq i^*$, using the product rule, the probability of following both right branchings of \cref{fig:f0} and continuing to the next iteration of the \textbf{for} loop is
\begin{align}
    \PrCont[i]
    &\leq 1 \cdot a_i=a_i.\label{eq:lateCont} 
\end{align}

\begin{figure}
\centering
\begin{subfigure}[t]{.45\textwidth}
\begin{tikzpicture}[scale=1]
`    \node {Round $i\geq i^*$, $f(x)=0$} 
        [sibling distance=2.5cm, level distance=2cm]
        child {node {return 1}
        edge from parent node [left,blue,text width=2cm,align=center] {$\leq a_i$ (\cref{eq:ai})}}
        child {
            child {node {return 0}
            edge from parent node [left,blue] {}}
            child {node {continue}
            edge from parent node [right,blue,text width=2cm,align=center] {$\leq a_i$ (\cref{eq:partCont})}}
        edge from parent node [right,blue] {$\leq 1$}};
    \end{tikzpicture}
    \caption{Probability tree diagram for $f(x)=0$, with bounds on probabilities of relevant event branches, with reference to the corresponding equations from the text where that bound is derived.}
    \label{fig:f0}
\end{subfigure}
\hfill
\begin{subfigure}[t]{.45\textwidth}
\begin{tikzpicture}[scale=1]
    \node {Round $i\geq i^*$, $f(x)=1$} 
        [sibling distance=2.5cm, level distance=2cm]
        child {node {return 1}
        edge from parent node [left,blue] {}}
        child {
            child {node {return 0}
            edge from parent node [left,blue,text width=2.3cm,align=center] {$\leq a_i$ \\($\sim$ \cref{eq:ai})}}
            child {node {continue}
            edge from parent node [right,blue] {$\leq 1$}}
        edge from parent node [right,blue,text width=2.3cm,align=center] {$\leq a_i$  \\($\sim$ \cref{eq:partCont})}};
    \end{tikzpicture}
    \caption{Probability tree diagram for $f(x)=1$, with bounds on probabilities of relevant event branches, derived using similar analyses as the equations referenced on the branches.}
    \label{fig:f1}
\end{subfigure}
\caption{Probability tree diagrams for a round of the \textbf{for} loop
 in \cref{alg:bool} when $i\geq i^*$, and $f(x)=0$ (\cref{fig:f0}), and 
 $f(x)=1$ (\cref{fig:f1}). By our choice of parameters, $a_i$ is small (it is
 always less than $1/4$), and decreases exponentially with increasing $i$.}
\end{figure}
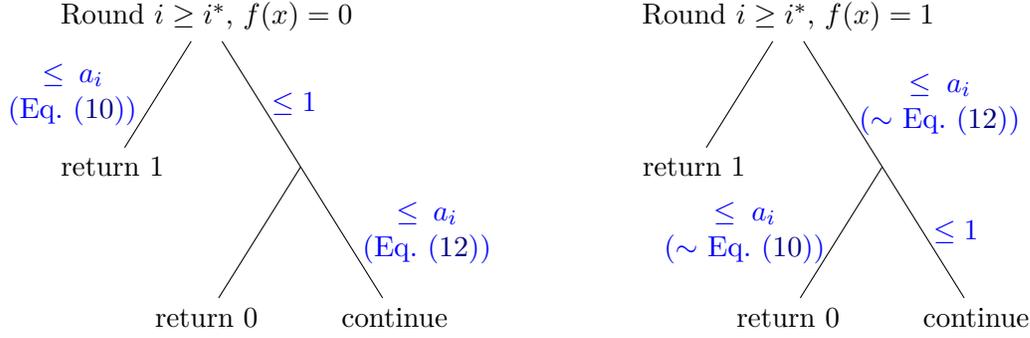

Finally, if we ever reach the end of the \textbf{for} loop without terminating, our algorithm returns 1, 
which is the wrong answer. This happens with probability 
\begin{align}
    Pr(final) &= \prod_{i=0}^{T} \PrCont \leq \prod_{i=i^*}^T a_i, \label{eq:prfinal} 
\end{align}
using \cref{eq: earlyCont}  for $i<i^*$ and \cref{eq:lateCont} for $i\geq i^*$.

Now we calculate the total probability of error. Plugging in \cref
{eq:ai}, \cref{eq: earlyCont}, \cref{eq:lateCont}, and 
\cref{eq:prfinal} into \cref{eq:probError1}, and splitting the first term of 
\cref{eq:probError1} into two parts to account for the different behavior of the
algorithm before and after round $i^*$, we get:
\begin{align}
    \label{eq:error} Pr(error) 
    & \leq \sum_{i=0}^{i^*}  a_i + \sum_{i=i^*+1}^{T}  
    \prod_{j = i^*}^i a_j + \prod_{i = i^*}^T a_i .
\end{align}

Since $N_i = 18(N - i)$, we have $a_i = e^{-(N-i)}$, which means the first sum in \cref{eq:error} is a geometric series, and is bounded by:
\begin{align}
    \sum_{i=0}^{i^*}  a_i &\leq  \frac{1}{e-1}\cdot e^{-N+(i^*+1)}<\delta/2,
\end{align}
where the final inequality arises from our choice of $N = \lceil T + \log(3/\delta)\rceil + 1$.
Combining the second and third terms of \cref{eq:error}, and upper bounding their $a_j$'s and $a_i$'s by $a_T$, we get another geometric series that sums to less than $\delta/2$:
\begin{align}
    \sum_{i=1}^{T-i^*+1} (a_T)^i < \dfrac{e^{-(N-T)}}{1-e^{-(N-T)}} < \delta/2.
\end{align}
Thus, $Pr(error) < \delta $, and our success probability is at least $1-\delta$.

Now we analyze the probability of error for $f(x)=1$ and set 
$i^*=\left\lceil\log\sqrt{3\wp{\sop P}{x}W_-(\sop P,f)}\right\rceil$. Then 
nearly identical analyses as in the $f(x)=0$ case (and using \cref{lem:SP_duals} to relate witness sizes of $\sop P$ and $\sop P^\dagger$) provide
the bounds on probabilities of relevant events, corresponding to branchings
in \cref{fig:f1}. By following the first right branching and then the next
left branching in \cref{fig:f1}, we see the probability of error at round $i$ for $i\geq i^*$
is
\begin{equation} \label{eq:ai_f1}
    \PrErr \leq a_i^2\leq a_i,
\end{equation}
since by our choice of parameters, $a_i$ is always less than $1$. Following the two right 
branchings in \cref{fig:f1}, the probability of continuing when $i\geq i^*$ is
\begin{align}
    \PrCont[j] 
    &\leq a_i\cdot 1=a_i.\label{eq:lateCont_f1} 
\end{align}
Thus the rest of the analysis is the same, and so we find that for 
$f(x)=0$ or $f(x)=1$, the probability of success is at least $1-\delta$.

\medskip

Now we calculate the average query complexity on input $x$, $\mathbb{E}[Q_x]$, given by 
\begin{equation}
    \mathbb{E}[Q_x]=\sum_{i=0}^{T} \left( \prod_{j=0}^{i-1} \PrCont[j] \right) \cdot (1-\PrCont[i]) \cdot Q(i).
    \label{eq:SP_avg_runtime}
\end{equation}
Here, $Q(i)$ is the number of queries used by the algorithm up to and including round $i$,
and $\left( \prod_{j=0}^{i-1} \PrCont[j] \right) \cdot (1-\PrCont[i])$ is the probability that we
terminate at round $i$.

The only time we make queries is in the Phase Checking subroutine. By \cref
{lem:phase_det}, the number of queries required to run a single repetition of
Phase Checking in the $i^\textrm{th}$ round is 
$O\left(\frac{1}{\theta} \log \left(\frac{1}{\epsilon}\right)\right) = O(\alpha_\pm \sqrt{W_\mp}) = O(2^i)$, since $\epsilon=\Theta(1)$. 
Taking into account the $N_i$ repetitions of Phase Checking
in the $i^\textrm{th}$ round, we find
\begin{equation} \label{eq:i1}
    Q(i) = \sum_{j=0}^{i} O\left(2^jN_j\right).
\end{equation}

Now setting $i^*$ to be $\left\lceil\log\sqrt{3\wp{\sop P}{x}W_-(\sop P,f)}\right\rceil$ or 
$\left\lceil\log\sqrt{3\wm{\sop P}{x}W_+(\sop P, f)}\right\rceil$ depending on whether $f(x)=1$ or $0$, respectively, we can use our bounds on event probabilities from our error analysis to bound the relevant probabilities for average query complexity. When $i < i^*$, we use the trivial bound $\PrCont[i]\leq 1$. When $i\geq i^*$,
we use \cref{eq:lateCont,eq:lateCont_f1} and our choice of $a_i$ to conclude
that $\PrCont[i]\leq a_i\leq 1/4$. For all $i$, we use that 
$(1-\PrCont[i])\leq 1.$ Splitting up the sum in \cref{eq:SP_avg_runtime} into $2$ terms,
for $i\leq i^*$ and $i>i^*$, and using these bounds on $\PrCont[i]$ along
with \cref{eq:i1}, we have
\begin{align}\label{eq:Q_intermediate}
     \mathbb{E}[Q_x]
     &\leq \sum_{i=0}^{i^*}  \sum_{j=0}^{i} O\left(2^jN_j \right)
    + \sum_{i=i^*+1}^T \left( \frac{1}{4} \right)^{i-i^*}  \sum_{j=0}^{i} O\left(2^jN_j \right).
\end{align}
We use the following inequalities to simplify \cref{eq:Q_intermediate},
\begin{align} \label{eq:magic_sum}
  &\sum_{j=0}^{i} 2^j (N-j)
                           \leq2^{i+1}(N-i+1),\quad \textrm{and}\\
 &\!\!\!\sum_{i=i^*+1}^{T}2^{-i}(T-i)\leq 2^{-i^*}(T-i^*)\label{eq:magic_sum2}
\end{align}
and finally find that
\begin{align}
    \mathbb{E}[Q_x] = O\left( 2^{i^*}(N-i^*+1) \right).
\end{align}

By our choice of $i^*$, $T$, and $N$, on input $x$ when $f(x)=1$,
the total query complexity is 
\begin{align}
\mathbb{E}[Q_x]=O\left(\sqrt{\wp{\sop P}{x}W_-(\sop P,f)}\log\left(\frac{W_+(\sop P,f)}{\wp{\sop P}{x} \delta}\right)\right),
\end{align}
and when $f(x)=0$,
the total average query complexity is 
\begin{align}
\mathbb{E}[Q_x]=O\left(\sqrt{\wm{\sop P}{x}W_+(\sop P,f)}\log\left(\frac{W_-(\sop P,f)}{\wm{\sop P}{x} \delta}\right)\right),
\end{align}
and the worst case query complexity is 
\begin{align} \label{eq:worst_case}
    \sum_{j=0}^{T} O\left(2^jN_j\right)=O\left(\sqrt{W_+W_-}\log(1/\delta)\right),
\end{align}
where we have again used \cref{eq:magic_sum}.
\end{proof}

\subsection{Application to st-connectivity}

As an example application of our algorithm, we analyze the query complexity of
$st$-connectivity on an $n$-vertex graph. There is a span program
$\mathcal{P}$ such that for inputs $x$ where there is a path from $s$
to $t$, $\wp{\mathcal{P}}{x}=R_{s,t}(x)$ where $R_{s,t}(x)$ is the effective
resistance from $s$ to $t$ on the subgraph induced by $x$,
 and for inputs $x$ where there is not a
path from $s$ to $t$, $\wm{\mathcal{P}}{x}=C_{s,t}(x)$, where $C_{s,t}(x)$ is the
effective capacitance between $s$ and $t$
\cite{belovsSpanProgramsQuantum2012,jarretQuantumAlgorithmsConnectivity2018}. In an $n$-vertex graph, the effective resistance is less than $n$, and the effective capacitance is less than $n^{2}$, so by \cref{thm:spEvalNew}, we can determine with bounded error that there is a path on input $x$ with $\tilde{O}(\sqrt{R_{s,t}(x)n^2})$ average queries or that there is not a path with
$\tilde{O}(\sqrt{C_{s,t}(x)n})$ average queries. In the worst case, when $R_{s,t}(x)=n$ or
$C_{s,t}(x)=n^2$, we recover the worst-case query complexity of $O(n^{3/2})$ of the original
span program algorithm.

The 
effective resistance is at most the shortest path between two vertices,
and the effective capacitance is at most the smallest cut between two
vertices. Thus
our algorithm determines whether or not there is a path from $s$ to
$t$ with $\tilde{O}(\sqrt{k}n)$ queries on average if there is a path of length $k$, and if there is no path,
the algorithm uses $\tilde{O}(\sqrt{cn})$ queries on average, where $c$ is the size of
the smallest cut between $s$ and $t$. Importantly, one does not need to know
bounds on $k$ or $c$ ahead of time to achieve this query complexity.

The analysis of the other examples listed in \cref{sec:intro} is similar.


\section{State Conversion Algorithm}\label{sec:StateConv}

Our main result for state conversion is the following: 

\begin{restatable}{theorem}{stateConvNew}\label{thm:stateConvNew}
Let $\mathscr{P}$ be a converting
vector set from $\{\ket{\rho_x}\}_{x\in X}$ to $\{\ket{\sigma_x}\}_{x\in X}$. Then there is a quantum algorithm such that for any $x\in X$, any failure probability $\delta\leq 1/3$, and any error $\varepsilon>0$,
\begin{enumerate}
    \item With probability $1-\delta$, on input $x$ the algorithm algorithm converts $\ket{\rho_x}$ to $\ket{\sigma_x}$ with error $\varepsilon$.
    \item On input $x$, if $w_+(\mathscr P,x)W_-(\mathscr P)\leq w_-(\mathscr P,x)W_+(\mathscr P)$, the average query complexity is
\begin{align}
O\left(\frac{\sqrt{\wp{\mathscr P}{x}
W_-(\mathscr P)}}{\varepsilon^{5}}\log\left(\frac{1}{\varepsilon}\right)\log\left(\frac{W_+(\mathscr P)}{\wp{\mathscr P}{x}\delta}\log\left(\frac{1}{\varepsilon}\right)\right)\right).
\end{align}

     If $w_+(\mathscr P,x)W_-(\mathscr P)\geq w_-(\mathscr P,x)W_+(\mathscr P)$, the average query complexity is
     \begin{equation}
O\left(\frac{\sqrt{\wm{\mathscr P}{x}
W_+(\mathscr P)}}{\varepsilon^{5}}\log\left(\frac{1}{\varepsilon}\right)\log\left(\frac{W_-(\mathscr P)}{\wm{\mathscr P}{x}\delta}\log\left(\frac{1}{\varepsilon}\right)\right)\right).
\end{equation}
\end{enumerate}
\end{restatable}

Comparing \cref{thm:stateConvNew} with \cref{cor:stateConv}, and considering the case of  $\varepsilon,\delta=\Omega(1)$, we see that in the worst case, when we have an input $x$ where $\wp{\mathscr P}{x}=W_+(\mathscr P)$ or $\wm{\mathscr P}{x}=W_-(\mathscr P)$
 the average query complexity of our algorithm is asymptotically the same
as the standard state conversion algorithm. However, when we have an instance $x$ with a smaller value of $w_\pm(\mathscr P,x)$,
then our algorithm has improved query complexity, without knowing anything about the input witness size ahead of time.

Our algorithm has worse scaling in $\varepsilon$ than \cref{cor:stateConv}, so
our algorithm will be most useful when $\varepsilon$
can be constant. One could also do a hybrid approach: initially run
our algorithm and then switch to that of \cref{cor:stateConv}. 

The problem of state conversion is a more general problem than function decision,
 and it can be used to solve the function decision problem.
However, because of the worse scaling with $\varepsilon$ in \cref
{thm:stateConvNew}, we considered function decision separately (see \cref
{sec:func}).

We use \cref{alg:StateConv} to prove \cref{thm:stateConvNew}. We now describe a key unitary, $\Up{\sop P}{x}{\alpha}{\hat{\varepsilon}}$, that appears in the algorithm.
In the following, we use most of the notation conventions of Ref.
\cite{leeQuantumQueryComplexity2011}. 
Let $\{\ket{\mu_i}\}_{i\in [q]}$ and $\{\ket{\nu_i}\}_{i\in [q]}$ be unit vectors in $\mathbb{C}^q$
as defined in \cite[Fact 2.4]{leeQuantumQueryComplexity2011},  such that 
\begin{equation}\label{eq:mu_nu}
\braket{\mu_i}{\nu_j}=\frac{k}{2(k-1)}(1-\delta_{i,j}).
\end{equation}For $X\subseteq [q]^n,$ let
$\mathscr{P}=\left(\{\ket{u_{xj}}\},\{\ket{v_{xj}}\}\right)$ where 
$\forall x\in X, j\in [n],\ket{v_{xj}},\ket{u_{xj}}\in\mathbb{C}^m$ be a converting
vector set from $\rho$ to $\sigma$. For all $x\in X$, the states $\ket{\rho_x}$ and $\ket{\sigma_x}$ are in the Hilbert space $\sop H$. Then for all $x\in X$, define 
$\ket{t_{x\pm}},
\ket{\psi_{x,\alpha,\hat{\varepsilon}}}\in(\mathbb{C}^2\otimes \sop H)\oplus(\mathbb{C}^n\otimes \mathbb{C}^q\otimes \mathbb{C}^m)$ as
\begin{align}\label{defeq:t_pm}
\ket{t_{x\pm}}=\frac{1}{\sqrt{2}}\left(\ket{0}\ket{\rho_x}\pm\ket{1}\ket{\sigma_x}\right), \quad \textrm{ and }\quad
\ket{\psi_{x,\alpha,\hat{\varepsilon}}}=\sqrt{\frac{\hat{\varepsilon}}{\alpha}}\ket{t_{x-}}-\sum_{j\in[n]}\ket{j}\ket{\mu_{x_j}}\ket{u_{xj}},
\end{align}
where $\alpha$ is
 analogous to the parameter $\alpha$ in \cref{eq:alphaIntro}. We will choose 
$\hat{\varepsilon}$ to achieve a desired
accuracy of $\varepsilon$ in our state conversion procedure. 
Set $\Lalp$ to equal the projection onto the orthogonal complement of the span
of the vectors $\{\ket{\psi_{x,\alpha,\hat{\varepsilon}}}\}_{x\in X}$, and set
$\PHxp=I-\sum_{j\in[n]}\proj{j}\otimes \proj{\mu_{x_j}}\otimes
I_{\mathbb{C}^n}$. Finally, we set $\Up{\sop
P}{x}{\alpha}{\hat{\varepsilon}}=(2\PHxp-I)(2\Lalp-I)$. The reflection $2\PHxp-I$ can be
implemented with two applications of $O_x$
\cite{leeQuantumQueryComplexity2011}, and the reflection $(2\Lalp-I)$ is independent of $x$ and so requires no queries. 

As with function decision, the 
time and query complexity of the algorithm is dominated by the number of applications of
$\Up{\sop P}{x}{\alpha}{\hat{\varepsilon}}$. If $T_U$ is the 
time required to implement $\Up{\sop P}{x}{\alpha}{\hat{\varepsilon}}$, then
the time complexity of our algorithm is simply the query complexity times $T_U$.


\begin{algorithm}[h!]
    \DontPrintSemicolon
    \SetKwInOut{Input}{Input}
    \SetKwInOut{Output}{Output}
    \Input{Converting vector set $\mathscr P$ from $\rho$ to $\sigma$, failure probability $\delta<1/3$, error $\varepsilon$, oracle $O_x$, initial state $\ket{\rho_x}$}
    \Output{$\ket{\tilde{\sigma}_x}$ such that $\|\ket{\tilde{\sigma}_x}-\ket{1}\ket{\sigma_x}\ket{0}\|\leq \varepsilon$}
    \tcc{Probing Stage}
    $\hat{\varepsilon}\gets \varepsilon^2/36$\; \label{algline:one}
    $T\gets \left\lceil\log \left(W_+(\mathscr P)W_-(\mathscr P)\right)\right\rceil$\;
    \For{i=0 \KwTo $T$}
    {
      $\delta_i\gets 2^{\log \delta-T+i-1}$
      
      \For{$\mathscr P'\in\{\mathscr P, \mathscr P^\dagger\}$}{
      $\alpha\gets 2^i/W_-(\mathscr P')$\;
      $\mathcal{A}\gets D(\Up{\mathscr P'}{x}{\alpha}{\hat{\varepsilon}})$ (\cref{lem:phase_det}) to precision $\hat{\varepsilon}^{3/2}/\sqrt{\alpha W_-(\mathscr P')}$ and accuracy $\hat{\varepsilon}^2$ \;              
        $\hat{a}\gets$ Amplitude Estimation (\cref{lem:ampEst}) of probability of outcome $\ket{0}_B$ in register $B$ when  $\mathcal{A}$ acts on $(\ket{0}\ket{\rho_x})_A\ket{0}_B$ to additive error $\hat{\varepsilon}/4$ with probability of failure $\delta_i$ \label{line:amp_Est_SC}\; 
        \lIf{$\hat{a}-1/2>-\frac{11}{4}\hat{\varepsilon}$}
        {
          Continue to State Conversion Stage
        }

      }

    }
    \tcc{State Conversion Stage}

      Apply $R(\Up{\mathscr P'}{x}{\alpha}{\hat{\varepsilon}})$ (\cref{lem:phase_refl}) with precision $\hat{\varepsilon}^{3/2}/\sqrt{\alpha W_-(\mathscr P')}$ and accuracy $\hat{\varepsilon}^2$ to $(\ket{0}\ket{\rho_x})_A\ket{0}_B$  and return the result \label{line:SC_conversion_stage}\;
    
    \caption{}
    \label{alg:StateConv}
\end{algorithm}{}


\noindent \textbf{High level idea of \cref{alg:StateConv}}: when we apply Phase Reflection of $\Up{\mathscr
 P'}{x}{\alpha}{\hat{\varepsilon}}$ (for $\mathscr P'\in\{\mathscr P, \mathscr P^\dagger\})$ in \cref{line:SC_conversion_stage} to
 $(\ket{0}\ket{\rho_x})_A\ket{0}_B=\frac{1}{\sqrt{2}}(\ket{t_{x+}}_A\ket{0}_B+\ket{t_{x-}}_A\ket{0}_B)$,
 we want $\frac{1}{\sqrt{2}}\ket{t_{x+}}_A\ket{0}_B$ to pick up a $+1$ phase and 
 $\frac{1}{\sqrt{2}}\ket{t_{x-}}_A\ket{0}_B$ to pick up a $-1$ phase. (Note that in this case, half of 
 the amplitude of the state is picking up a $+1$ phase, and half is picking up a $-1$ phase.) If this were to happen perfectly,
 we would have the desired state $(\ket{1}\ket{\sigma_x})_A\ket{0}_B=\frac{1}{\sqrt{2}}(\ket{t_{x+}}_A\ket{0}_B
 -\ket{t_{x-}}_A\ket{0}_B)$. We show that if $\alpha$ is larger than a critical value that depends on the witness size of the input $x$, then in \cref{line:SC_conversion_stage}, we will mostly pick up the desired phase.
However, we don't know ahead of time how large $\alpha$ should be. To
determine this, we implement the Probing Stage (Lines 1-9), which uses Amplitude
Estimation of a Phase Checking subroutine to test exponentially increasing values of $\alpha$. 

We use the following two Lemmas (\cref{claim:stopping} and 
\cref{lem:finalStateAnalysis}) to analyze 
\cref{alg:StateConv} and prove \cref{thm:stateConvNew}:

\begin{restatable}{lemma}{half}\label{claim:stopping}
For a converting vector set $\mathscr P$ that converts $\rho$ to $\sigma$, and Phase Checking of $\Up{\mathscr
 P}{x}{\alpha}{\hat{\varepsilon}}$ done with accuracy $\hat{\varepsilon}^2$ and 
precision $\Theta = \hat{\varepsilon}^{3/2}/\sqrt{\alpha W_-(\mathscr P)}$, then 
\begin{enumerate}
\item If $\alpha\geq \wp{\mathscr P}{x}$, then  $\|\Pi_0(\Up{\mathscr
 P}{x}{\alpha}{\hat{\varepsilon}})\ket{t_{x+}}\ket{0}\|^2> 1-\hat{\varepsilon}.$ \label{claim:stopping_1}
\item If $\alpha\geq 1/W_-(\mathscr P)$, then $\| P_\Theta(\Up{\mathscr
 P}{x}{\alpha}{\hat{\varepsilon}})\ket{t_{x-}} \|^2 \leq \frac{\hat{\varepsilon}^2}{2}$ and $\left\|\Pi_0(\Up{\mathscr
 P}{x}{\alpha}{\hat{\varepsilon}})\ket{t_{x-}}\ket{0}\right\|\leq \hat{\varepsilon}(1+1/\sqrt{2})$.\label{claim:stopping_2}
\end{enumerate}
\end{restatable}

\cref{claim:stopping} \cref{claim:stopping_2} ensures that the $\ket{t_{x-}}$ part of the state mostly
picks up a $-1$ phase when we apply Phase Reflection regardless of the value of $\alpha$, and \cref{claim:stopping} \cref{claim:stopping_1} ensures that when $\alpha$ is large enough, the $\ket{t_{x+}}$ part of the state mostly
picks up a $+1$ phase.
\cref{claim:stopping} plays a similar role in state conversion to \cref{lem:phase_est_early} in function decision. It shows us that
the behavior of the algorithm changes at some point when $\alpha$ is large enough,
without our having to know $\alpha$ ahead of time (\cref{claim:stopping_1}) but it 
also is used to show that we don't terminate early when we shouldn't, leading to an
incorrect outcome (\cref{claim:stopping_2}). 

\begin{proof}[Proof of \cref{claim:stopping}]
Throughout this proof, $P_\Theta$ and $\Pi_0$ are shorthand 
for $P_\Theta(\Up{\mathscr{P}}{x}{\alpha}{\hat{\varepsilon}})$ and 
$\Pi_0(\Up{\mathscr{P}}{x}{\alpha}{\hat{\varepsilon}})$, with Phase Checking done
to precision $\Theta=\hat{\varepsilon}^{3/2}/\sqrt{\alpha W_-(\mathscr P)}$ and accuracy 
$\hat{\varepsilon}^2$.

\vspace{.5cm}
\textit{Part 1:}
We first prove that
$\|P_0\ket{t_{x+}}\|^2\geq 1-\hat{\varepsilon}$, which by \cref{lem:phase_det} gives us a bound on $\|\Pi_0\ket{t_{x+}}\|^2$. Following \cite[Claim 4.4]{leeQuantumQueryComplexity2011}, we consider the state
\begin{equation}
\ket{\varphi}=\ket{t_{x+}}+\frac{\sqrt{\hat{\varepsilon}}}{2\sqrt{\alpha}}\frac{2(k-1)}{k}\sum_{j\in[n]}\ket{j}\ket{\nu_{x_j}}\ket{v_{xj}},
\end{equation}
where $\ket{\nu_{x_j}}$ is from \cref{eq:mu_nu}.
Note that for $\ket{\psi_{y,\alpha,\hat{\varepsilon}}}$ for all $y\in X$, because
$\braket{t_{y-}}{t_{x+}}=\frac{1}{2}\left(\braket{\rho_y}{\rho_x}-\braket{\sigma_y}{\sigma_x}\right)$,
and also $\sum_{j:x_j\neq
y_j}\braket{u_{yj}}{v_{xj}}=\braket{\rho_y}{\rho_x}-\braket{\sigma_y}{\sigma_x}$
from the constraints of \cref{eq:filteredNorm}, we have
\begin{equation}
\braket{\psi_{y,\alpha,\hat{\varepsilon}}}{\varphi}=\frac{\sqrt{\hat{\varepsilon}}}{\sqrt{\alpha}}\braket{t_{y-}}{t_{x+}}-\frac{\sqrt{\hat{\varepsilon}}}{2\sqrt{\alpha}}\sum_{j:x_j\neq y_j}\braket{u_{yj}}{v_{xj}}=0.
\end{equation}

Because $\ket{\varphi}$ is orthogonal to all of the $\ket{\psi_{y,\alpha,\hat{\varepsilon}}}$, we have
$\Lalp\ket{\varphi}=\ket{\varphi}$. Also, $\PHxp\ket{\varphi}=\ket{\varphi}$ since
$\PHxp\ket{t_{x+}}=\ket{t_{x+}}$ and $\braket{\mu_{x_j}}{\nu_{x_j}}=0$ for every $j$.
Thus $P_0\ket{\varphi}=\ket{\varphi}$. 

Note
\begin{equation}
\braket{\varphi}{\varphi}=1+\frac{\hat{\varepsilon}}{4\alpha}\frac{4(k-1)^2}{k^2}w_+(\mathscr P,x)< 1+\hat{\varepsilon},
\end{equation}
because of our assumption that $\alpha\geq w_+(\mathscr P,x)$. Also, $\braket{t_{x+}}{\varphi}=1$,
so 
\begin{equation}
\|P_0\ket{t_{x+}}\|^2\geq \|\proj{\varphi}\ket{t_{x+}}\|^2/\|\ket{\varphi}\|^4\geq \frac{1}{1+\hat{\varepsilon}}> 1-\hat{\varepsilon}.
\end{equation}
Then by \cref{lem:phase_det}, we have $\|P_0\ket{t_{x+}}\|^2\leq \|\Pi_0\ket{t_{x+}}\ket{0}\|^2$, so 
\begin{equation}\label{eq:lined}
\|\Pi_0\ket{t_{x+}}\ket{0}\|^2> 1-\hat{\varepsilon}.
\end{equation}

\vspace{.5cm}
\textit{Part 2:}
Let $\ket{w} = \sqrt{\frac{\alpha}{\hat{\varepsilon}}} \ket{\psi_{x,\alpha,\hat{\varepsilon}}}$, 
 so $\Lalp \ket{w} = 0$ and $\PHxp \ket{w} = \ket{t_{x-}}$.
Applying \cref{spec_gap_lemm}, we have
 \begin{equation}\label{eq:specGap1}
\| P_\Theta \ket{t_{x-}} \|^2 = \| P_\Theta \PHxp\ket{w} \|^2 \leq \frac{\Theta^2 }{4}\| \ket{w}\|^2.
 \end{equation}
Now
\begin{equation}\label{eq:specGap2}
\left\| \ket{w}\right\|^2 = \frac{\alpha}{\hat{\varepsilon}}\left( \frac{\hat{\varepsilon}}{\alpha}\iner{t_{x-}} + \Sigma_j \iner{u_{xj}}\right) \leq 1 + \frac{\alpha W_-}{\hat{\varepsilon}}.
\end{equation}
Combining \cref{eq:specGap1,eq:specGap2}, and setting $\Theta = \hat{\varepsilon}^{3/2}/\sqrt{\alpha W_-}$, we have that
\begin{equation} \label{eq:Ptheta_tx-}
\| P_\Theta \ket{t_{x-}} \|^2\leq \frac{\hat{\varepsilon}^3}{4\alpha W_-}\left( 1 + \frac{\alpha W_-}{\hat{\varepsilon}} \right) \leq \frac{\hat{\varepsilon}^2}{2},
\end{equation}
where we have used our assumption from the statement of the lemma that $\alpha \geq 1/W_-$.

We next bound $\|\Pi_0\ket{t_{x-}})\ket{0}\|.$
Inserting the identity operator $I=P_{\Theta}+\overline{P}_{\Theta}$ for $\Theta=\hat{\varepsilon}^{3/2}/\sqrt{\alpha W_-}$, we have
\begin{align}
 \|\Pi_{0}\ket{t_{x-}}\ket{0}\|
& =\left\|\Pi_{0}\left(\left(P_{\Theta}+\overline{P}_{\Theta}\right)\ket{t_{x-}}\right)\ket{0}\right\|\\
&\leq \|P_{\Theta}\ket{t_{x-}}\|+\|\Pi_{0}\left(\overline{P}_{\Theta}\ket{t_{x-}}\right)\ket{0}\|\\
&\leq \frac{\hat{\varepsilon}}{\sqrt{2}}+\hat{\varepsilon}.
\end{align}
where the second line comes from the triangle inequality and the fact that a projector acting on a vector cannot increase its norm.
The first term in the final line comes from \cref{eq:Ptheta_tx-}, and the second term comes \cref{lem:phase_det}.
\end{proof}

The following lemma, \cref{lem:finalStateAnalysis}, tells us that when we break out of the Probing Stage due to
a successful Amplitude Estimation in \cref{line:amp_Est_SC}, we will convert $\ket{\rho_x}$
to $\ket{\sigma_x}$ with appropriate error in the State Conversion Stage in \cref{line:SC_conversion_stage}, regardless of the value of $\alpha$ (\cref{item:break_out} in \cref{lem:finalStateAnalysis}). However, \cref{lem:finalStateAnalysis} also tells us that once $\alpha\geq \wp{\mathscr P}{x}$,
then if Amplitude Estimation does not fail, we will exit the Probing Stage (\cref{item:large_alpha} in \cref{lem:finalStateAnalysis}). Together \cref{item:break_out} and \cref{item:large_alpha} ensure that once $\alpha$ is large enough, the algorithm will be very likely to terminate and correctly produce the output state, but before $\alpha$ is large enough, if there is some additional structure in the converting
vector set that causes our Probing Stage to end early (when $\alpha< \wp{\mathscr P}{x}$), we will still have a successful result.

\begin{restatable}{lemma}{finalState}\label{lem:finalStateAnalysis}
For a converting vector set $\mathscr P$ that converts $\rho$ to $\sigma$, and Phase Checking and Phase Reflection of $\Up{\mathscr
 P}{x}{\alpha}{\hat{\varepsilon}}$ done with accuracy $\hat{\varepsilon}^2$ and 
precision $\Theta = \hat{\varepsilon}^{3/2}/\sqrt{\alpha W_-(\mathscr P)}$ 
for $\hat{\varepsilon}\leq1/9$, and $\alpha\geq 1/W_-(\mathscr{P}),$
\begin{enumerate}
\item If 
\begin{equation}\label{eq:condition_breaking}
\|\Pi_0(\Up{\mathscr P}{x}{\alpha}{\hat{\varepsilon}})\ket{0}\ket{\rho_x}_A\ket{0}_B\|^2> \frac{1}{2}-3\hat{\varepsilon},
\end{equation}
then 
 \begin{equation}\label{eq:SC_success}
 \|R(\Up{\mathscr{P}}{x}{\alpha}{\hat{\varepsilon}})(\ket{0}\ket{\rho_x})_A\ket{0}_B-(\ket{1}\ket{\sigma_x})_A\ket{0}_B\|\leq 6\sqrt{\hat{\varepsilon}}.
 \end{equation} \label{item:break_out}
 \item If $\alpha\geq \wp{\mathscr P}{x}$
then $\|\Pi_0(\Up{\mathscr P}{x}{\alpha}{\hat{\varepsilon}})(\ket{0}\ket{\rho_x})_A\ket{0}_B\|^2> \frac{1}{2}-\frac{5\hat{\varepsilon}}{2}$.
\label{item:large_alpha}
\end{enumerate}
\end{restatable}

\begin{proof}
~\\
For the rest of the proof, we will use $P_\Theta$, $\Pi_0$, 
and $R$ as shorthand for $P_\Theta(\Up{\mathscr{P}}{x}{\alpha}{\hat{\varepsilon}})$, 
$\Pi_0(\Up{\mathscr{P}}{x}{\alpha}{\hat{\varepsilon}})$, and
 $R(\Up{\mathscr{P}}{x}{\alpha}{\hat{\varepsilon}})$ respectively.

\noindent\textit{Part 1:} We have 
\begin{align}
\|R\ket{0}\ket{\rho_x}\ket{0}-\ket{1}\ket{\sigma_x}\ket{0}\|=
\frac{1}{\sqrt{2}}\|R(\ket{t_{x+}}+\ket{t_{x-}})\ket{0}-(\ket{t_{x+}}-\ket{t_{x-}})\ket{0}\|\nonumber\\
\leq \frac{1}{\sqrt{2}}\|(R-I)\ket{t_{x+}}\ket{0}\|+\frac{1}{\sqrt{2}}\|(R+I)\ket{t_{x-}}\ket{0}\|.\label{eq:errorBoundSC1}
\end{align}

In the first term of \cref{eq:errorBoundSC1}, we can replace $R$ with $\Pi_0-\overline{\Pi}_0$ (as described above \cref{lem:phase_refl}), and we can insert $I=\Pi_0+\overline{\Pi}_0$ to get
\begin{equation}\label{eq:equation_interrupted}
\frac{1}{\sqrt{2}}\|(R-I)\ket{t_{x+}}\ket{0}\|=\frac{1}{\sqrt{2}}\|((\Pi_0-\overline{\Pi}_0)-I)(\Pi_0+\overline{\Pi}_0)\ket{t_{x+}}\ket{0}\|=\frac{2}{\sqrt{2}}\|\overline{\Pi}_0\ket{t_{x+}}\ket{0}\|,
\end{equation}
where we have used the fact that $\Pi_0$ and $\overline{\Pi}_0$ are orthogonal.

To bound $\frac{2}{\sqrt{2}}\|\overline{\Pi}_0\ket{t_{x+}}\ket{0}\|$, we start from our assumption that $\|\Pi_0(\ket{0}\ket{\rho_x})_A\ket{0}_B\|^2\geq 1/2-3\hat{\varepsilon}$. Writing $\ket{0}\ket{\rho_x}$ in terms of $\ket{t_{x+}}$ and $\ket{t_{x+}}$, and using the triangle inequality, we have
\begin{equation}
\frac{1}{2}-3\hat{\varepsilon}\leq \frac{1}{2}\|\Pi_0(\ket{t_{x+}}+\ket{t_{x-}})\ket{0}\|^2
\leq 
\frac{1}{2}\left(\|\Pi_0\ket{t_{x+}}\ket{0}\|+\|\Pi_0\ket{t_{x-}}\ket{0}\|\right)^2\label{line1}.
\end{equation}

From \cref{claim:stopping} \cref{claim:stopping_2} , we have $\|\Pi_0\ket{t_{x-}}\ket{0}\|<2\hat{\varepsilon}$, so plugging into \cref{line1}, we have
\begin{equation}
1-6\hat{\varepsilon}\leq\left(\|\Pi_0\ket{t_{x+}}\ket{0}\|+2\hat{\varepsilon}\right)^2.
\end{equation}
Rearranging, we find:
\begin{equation}
(\sqrt{1-6\hat{\varepsilon}}-2\hat{\varepsilon})^2\leq\|\Pi_0\ket{t_{x+}}\ket{0}\|^2.
\end{equation}
Since $\|\Pi_0\ket{t_{x+}}\ket{0}\|^2+\|\overline{\Pi}_0\ket{t_{x+}}\ket{0}\|^2=1$, we have
\begin{equation}
 \|\overline{\Pi}_0\ket{t_{x+}}\ket{0}\|^2\leq 1-(\sqrt{1-6\hat{\varepsilon}}-2\hat{\varepsilon})^2<10\hat{\varepsilon}.
\end{equation} 
Plugging back into \cref{eq:equation_interrupted}, we find
\begin{equation}
\frac{1}{\sqrt{2}}\|(R-I)\ket{t_{x+}}\ket{0}\|=\frac{2}{\sqrt{2}}\|\overline{\Pi}_0\ket{t_{x+}}\ket{0}\|\leq 2\sqrt{5\hat{\varepsilon}}.
\label{eq:RItplus}
\end{equation}

In the second term of \cref{eq:errorBoundSC1}, we again replace $R$ with 
$\Pi_0-\overline{\Pi}_0$ and replace $I$ with $\Pi_0+\overline{\Pi}_0$ to get
\begin{align}
\frac{1}{\sqrt{2}}\|(R+I)\ket{t_{x-}}\ket{0}\|=
\frac{1}{\sqrt{2}}\|(\Pi_0-\overline{\Pi}_0+\Pi_0+\overline{\Pi}_0)\ket{t_{x-}}\ket{0}\|=\frac{2}{\sqrt{2}}\|\Pi_0\ket{t_{x-}}\ket{0}\|\leq \hat{\varepsilon}(1+1/\sqrt{2})
\label{eq:RItminus_new}
\end{align}
by \cref{claim:stopping} \cref{claim:stopping_2}, where we have used our assumption that
 $\alpha\geq 1/W_-(\mathscr{P})$.

 Combining \cref{eq:RItplus,eq:RItminus_new}, and plugging into into 
\cref{eq:errorBoundSC1} and using that $\hat{\varepsilon}\leq 1/9$, we have
\begin{equation}
\|R\ket{0}\ket{\rho_x}\ket{0}-\ket{1}\ket{\sigma_x}\ket{0}\|
\leq 2\sqrt{5\hat{\varepsilon}}+\hat{\varepsilon}(1+1/\sqrt{2})<6\sqrt{\hat{\varepsilon}}.
\end{equation}

\vspace{.5cm}
\noindent\textit{Part 2:} We analyze $\|\Pi_0{\ket{0}\ket{\rho_x}}\ket{0}\|.$
Using the triangle inequality, we have
\begin{equation}\label{eq:linef}
\|\Pi_0{\ket{0}\ket{\rho_x}}\ket{0}\|\geq\frac{1}{\sqrt{2}}\left(
\|\Pi_0\ket{t_{x+}}\ket{0}\|-\left\|\Pi_0\ket{t_{x-}}\ket{0}\right\|\right).
\end{equation}
The first term we bound using \cref{claim:stopping} \cref{claim:stopping_1} and the second with \cref{claim:stopping} \cref{claim:stopping_2} to give us
\begin{equation}
\|\Pi_0{\ket{0}\ket{\rho_x}}\ket{0}\|\geq \frac{1}{\sqrt{2}}\left(\sqrt{1-\hat{\varepsilon}}-\hat{\varepsilon}\left(1+\frac{1}{\sqrt{2}}\right)\right)
\end{equation}
Squaring both sides and using a series expansion, we find
\begin{equation}
\|\Pi_0{\ket{0}\ket{\rho_x}}\ket{0}\|^2> \frac{1}{2}-\left(\frac{3}{2}+\frac{1}{\sqrt{2}}\right)\hat{\varepsilon}>\frac{1}{2}-\frac{5}{2}\hat{\varepsilon}.
\label{eq:breakout_cond}
\end{equation}
\end{proof}

With \cref{claim:stopping} and \cref{lem:finalStateAnalysis}, we can now proceed
to the proof of \cref{thm:stateConvNew}:

\begin{proof}[Proof of \cref{thm:stateConvNew}]
We analyze \cref{alg:StateConv}. Note that the $l_2$ norm between two quantum states is at most $2$, so we may assume $\varepsilon\leq 2$ and hence $\hat{\varepsilon}\leq 9.$

First we show that the probability of returning a state $\ket{\tilde{\sigma}_x}$ 
such that $\| \ket{\tilde{\sigma}_x} - \ket{1}\ket{\sigma_x}\ket{0} \| \geq \varepsilon$ is at most $\delta$. We first analyze the case that $w_+(\mathscr P,x)W_-(\mathscr P)\leq w_-(\mathscr P,x)W_+(\mathscr P)=w_+(\mathscr P^\dagger,x)W_-(\mathscr P^\dagger)$.
Let $i^* =\lceil \log(w_+(\mathscr P,x)W_-(\mathscr P))\rceil$ 
and $T = \lceil\log W_+(\mathscr P)W_-(\mathscr P)\rceil$.

Notice that once $i\geq i^*$, we have $\alpha\geq \wp{\mathscr P}{x}$, so by 
\cref{lem:finalStateAnalysis} \cref{item:large_alpha} we have
$\|\Pi_0({\ket{0}\ket{\rho_x}})_A\ket{0}_B\|^2\geq \frac{1}{2}-\frac{5\hat{\varepsilon}}{2}$. Thus when we do Amplitude Estimation in \cref{line:amp_Est_SC} of 
\cref{alg:StateConv} to additive error $\hat{\varepsilon}/4$, with probability $1-\delta_{i}$ we will find the probability of outcome $\ket{0}_B$ will be at least
$1/2-11\hat{\varepsilon}/4$, causing us to continue to the State Conversion Stage.
Furthermore, by combining \cref{lem:finalStateAnalysis} \cref{item:break_out} and \cref{item:large_alpha}, 
our algorithm is guaranteed to output the target state within error $6\sqrt{\hat{\varepsilon}}=\varepsilon$
regardless of an error in Amplitude Estimation. Therefore, the algorithm can only return a wrong state before
round $i^*$, and only if Amplitude Estimation fails. 
Thus, we calculate the probability of error as:
\begin{equation}
    Pr(error) = \sum_{i=0}^{i^*-1} \left( \prod_{j=0}^{i-1} \PrCont[j] \right) \cdot \PrErr
\end{equation}
where $\PrCont$ is the
probability of continuing to the next round of the \textbf{for} loop at round
$i$,
and $\PrErr$ is the probability of a failure of Amplitude Estimation at round $i$, both conditioned on reaching round $i$. We upper bound $\PrCont[j]$ by 1 and $\PrErr$ by $2\delta_i$, as $\delta_i$ is the probability of Amplitude Estimation failure in \cref{line:amp_Est_SC}, and we do two rounds (for $\mathscr P$ and $\mathscr P^\dagger$). This then gives us:
\begin{align}
    Pr(error) 
    &\leq \sum_{i=0}^{i^*-1} 2\delta_i \nonumber\\ 
    &= \sum_{i=0}^{i^*-1} 2\delta\cdot 2^{-T +i-1} \nonumber\\
    &\leq\delta,
\end{align}
where we have used our choice of $\delta_i$ and that $i^*\leq T$. Thus the probability of error is bounded by 
$\delta$. 

Now we analyze the average query complexity.
Let $Q(i)$ be the query complexity of the algorithm when it exits the Probing Stage at round $i$.
Then the average query complexity on input $x$ is
\begin{equation}
    \mathbb{E}[Q_x]=\sum_{i=0}^T \left( \prod_{j=0}^{i-1}\PrCont[j] \right)\left(1-\PrCont[i]\right) Q(i),
\end{equation}
where $\left( \prod_{j=0}^{i-1}\PrCont[j] \right)\left(1-\PrCont[i]\right)$ is the probability of terminating at round $i$.

At the $i$th round of the Probing Stage, we implement
Phase Checking with precision $\hat{\varepsilon}^{3/2}/\sqrt{2^i}$ and
accuracy $\hat{\varepsilon}^2$, which uses $O\left(\frac{\sqrt{2^i
}}{\hat{\varepsilon}^{3/2}}\log\left(\frac{1}{\hat{\varepsilon}}\right)\right)$ queries for a single iteration. 
By \cref{lem:ampEst}, we use
$O\left(\frac{1}{\hat{\varepsilon}}\log\left(\frac{1}{\delta_i}\log\left(\frac{1}{\hat{\varepsilon}}\right)\right)\right)$ 
applications of Phase Checking inside the Iterative Amplitude Estimation subroutine to reach success probability of at least
$1-\delta_i$ with error $\hat{\varepsilon}/4$. 
Therefore,
\begin{align}
Q(i)&=\left(\sum_{j=0}^i O\left(\frac{\sqrt{2^j
}}{\hat{\varepsilon}^{5/2}}\log\left(\frac{1}{\hat{\varepsilon}}\right) \log\left(\frac{1}{\delta_j}\log\left(\frac{1}{\hat{\varepsilon}}\right)\right)\right)
\right)+O\left(\frac{\sqrt{2^i
}}{\hat{\varepsilon}^{3/2}}\log\left(\frac{1}{\hat{\varepsilon}}\right)\right)\nonumber\\
&=\left(\sum_{j=0}^i O\left(\frac{\sqrt{2^j
}}{\hat{\varepsilon}^{5/2}}\log\left(\frac{1}{\hat{\varepsilon}}\right) \log\left(\frac{1}{\delta_j}\log\left(\frac{1}{\hat{\varepsilon}}\right)\right)\right)
\right),
\end{align}
where the second term in the first line is the query complexity of the State Conversion Stage,
so the complexity is dominated by the Probing Stage.

We divide the analysis into two parts: $i\leq i^*$ and $i^*<i\leq T$.
When $i \leq i^*$, we use the trivial bound $\left( \prod_{j=0}^{i-1} \PrCont[j] \right)\left(1-\PrCont[i]\right) \leq 1$.
Thus, the contribution to the average query complexity from rounds $i$ with $i \leq i^*$ is at most:
\begin{align}\label{eq:total_q_early}
    \sum_{i=0}^{i^*} Q(i)&\leq \sum_{i=0}^{i^*} \sum_{j=0}^iO\left(\frac{\sqrt{2^j
}}{\hat{\varepsilon}^{5/2}}\log\left(\frac{1}{\hat{\varepsilon}}\right) \log\left(\frac{1}{\delta_j}\log\left(\frac{1}{\hat{\varepsilon}}\right)\right)\right) \\\nonumber
&=\hat{\varepsilon}^{-5/2}\log\left(\frac{1}{\hat{\varepsilon}}\right)\sum_{i=0}^{i^*} \sum_{j=0}^i O\left(\sqrt{2^j}\left(T- j +\log\left(\frac{1}{\delta}\log\left(\frac{1}{\hat{\varepsilon}}\right)\right) +1\right)\right) \\\nonumber
    &=O\left(\frac{\sqrt{\wp{\mathscr P}{x}
W_-(\mathscr P)}}{\hat{\varepsilon}^{5/2}}\log\left(\frac{1}{\hat{\varepsilon}}\right)\log\left(\frac{W_+(\mathscr P)}{\wp{\mathscr P}{x}\delta}\log\left(\frac{1}{\hat{\varepsilon}}\right)\right)\right)
\end{align}
where we have used the following inequality twice:
\begin{equation}\label{eq:inequality2}
    \sum_{j=0}^i \sqrt{2^j}(\log W-j) \; \leq\;4\sqrt{2^i}(\log W-i+3).
\end{equation}

For $i $ from $i^*+1$ to $T$, as discussed below \cref{eq:breakout_cond}, Amplitude Estimation in Line 8 should
produce an estimate that triggers breaking out of the Probing Stage at Line 9. Thus the 
probability of continuing to the  next iteration depends on Amplitude Estimation failing, 
which happens with probability $\delta_i\leq \delta\leq\frac{1}{2\sqrt{2}}$, since $\delta\leq 1/3$. Using $(1-\PrCont[i])\leq 1$, we thus have
\begin{align}
\left( \prod_{j=0}^{i-1}\PrCont[j] \right)\left(1-\PrCont[i]\right)\leq \left(\frac{1}{2\sqrt{2}}\right)^{i-i^*}.
\end{align}
The contribution to the average query complexity for rounds after $i^*$ is therefore
\begin{align}\label{eq:total_Q_late}
    \sum_{i=i^*+1}^{T} \left(\frac{1}{2\sqrt{2}}\right)^{i-i^*}Q(i) 
    &= \sum_{i=i^*+1}^{T} \left(\frac{1}{2\sqrt{2}}\right)^{i-i^*} \sum_{j=0}^i 
    O\left(\frac{\sqrt{2^j}}{\hat{\varepsilon}^{5/2}}\log\left(\frac{1}{\hat{\varepsilon}}\right)
    \log\left(\frac{1}{\delta_j}\log\left(\frac{1}{\hat{\varepsilon}}\right)\right)\right)\nonumber \\
&= O\left(\frac{\sqrt{\wp{\mathscr P}{x}
W_-(\mathscr P)}}{\hat{\varepsilon}^{5/2}}\log\left(\frac{1}{\hat{\varepsilon}}\right)\log\left(\frac{W_+(\mathscr P)}{\wp{\mathscr P}{x}\delta}\log\left(\frac{1}{\hat{\varepsilon}}\right)\right)\right)
\end{align}
where we have used \cref{eq:inequality2} and \cref{eq:magic_sum2}. 

Combining \cref{eq:total_q_early,eq:total_Q_late} and replacing $\hat{\varepsilon}$ with $\varepsilon$, the average query complexity of the algorithm on input $x$ is
\begin{align}
\mathbb{E}[Q_x]=O\left(\frac{\sqrt{\wp{\mathscr P}{x}
W_-(\mathscr P)}}{\varepsilon^{5}}\log\left(\frac{1}{\varepsilon}\right)\log\left(\frac{W_+(\mathscr P)}{\wp{\mathscr P}{x}\delta}\log\left(\frac{1}{\varepsilon}\right)\right)\right).
\label{eq:exactStateConv}
\end{align}

When $w_+(\mathscr P,x)W_-(\mathscr P)\geq w_-(\mathscr P,x)W_+(\mathscr P)=w_+(\mathscr P^\dagger,x)W_-(\mathscr P^\dagger)$, using the same analysis but with $\mathscr P^\dagger$ and applying \cref{lem:SP_duals}, we find 
\begin{equation}
\mathbb{E}[Q_x]=O\left(\frac{\sqrt{\wm{\mathscr P}{x}
W_+(\mathscr P)}}{\varepsilon^{5}}\log\left(\frac{1}{\varepsilon}\right)\log\left(\frac{W_-(\mathscr P)}{\wm{\mathscr P}{x}\delta}\log\left(\frac{1}{\varepsilon}\right)\right)\right).
\end{equation}

\end{proof}

\subsection{Function Evaluation with Fast Verification}

The state conversion algorithm can be used to evaluate a discrete function $f:X\rightarrow [m]$ for $X\subseteq [q]^n$ on input $x$ by converting from $\ket{\rho_x}=\ket{0}$ to $\ket{\sigma_x}=\ket{f(x)}$ 
and then measuring in the standard basis to learn $f(x)$.
When the correctness of $f(x)$ can be verified with an additional constant number
of queries, we can modify our state conversion algorithm to remove the Probing Stage,
and instead use the correctness verification of the output state as a test of whether
the algorithm is complete. In this case, we can remove a log factor from 
the complexity:
\begin{restatable}{theorem}{verifyStateConversion}\label{thm:verify_state_conversion}
For a function $f:X\rightarrow [m]$, such that $f(x)$ can be verified without error
using at most constant additional queries to $O_x$, given a converting vector set $\mathscr P$
from $\rho=\{\ket{0}\}_{x\in X}$ to $\sigma=\{\ket{f(x)}\}_{x\in X}$ and $\delta<2^{-1/2}$, then there is a quantum algorithm
that correctly evaluates $f$ with probability at least $1-\delta$
and uses
\begin{equation}
O\left(\frac{\sqrt{\min\{\wp{\mathscr P}{x}W_-(\mathscr P),\wm{\mathscr P}{x}W_+(\mathscr P)\}}}{\delta^{3/2}}\log\left(\frac{1}{\delta}\right)\right)
\end{equation}
average queries on input $x.$
\end{restatable}

While removing a log factor might seem inconsequential, it yields
an exponential quantum advantage in the next section for some applications, as opposed to only
a superpolynomial advantage. 

\begin{proof}[Proof of \cref{thm:verify_state_conversion}]
We analyze \cref{alg:StateConvVerify}, which is similar to \cref{alg:StateConv}, 
but with the Probing Stage replaced by a post-State Conversion verification procedure. 

\begin{algorithm}[h!]
    \DontPrintSemicolon
    \SetKwInOut{Input}{Input}
    \SetKwInOut{Output}{Output}
    \Input{Oracle $O_x$ for $x\in X\subseteq[q]^n$, function $f:X\rightarrow[m]$ for $m,q\in \mathbb{N}$, a converting vector set $\mathscr P$ from $\rho$ to $\sigma$, where
    $\ket{\sigma_x}=\ket{f(x)}$, a procedure to verify the correctness of $f(x)$
    with constant additional queries to $O_x$. Probability of error $\delta<2^{-1/2}$, initial state $\ket{\rho_x}$}
    \Output{$f(x)$ with probability $1-\delta$}
    $\hat{\varepsilon}\gets \delta/36$\; \label{algline:oneA}
    $T\gets \left\lceil\log \left(W_+(\mathscr P)W_-(\mathscr P)\right)\right\rceil$\;
    \For{i=0 \KwTo $T$}
    {
      
      \For{$\mathscr P'\in\{\mathscr P, \mathscr P^\dagger\}$}{
      $\alpha\gets 2^i/W_-(\mathscr P')$\;
       $\ket{\psi}\gets$ Apply $R(\Up{\mathscr P'}{x}{\alpha}{\hat{\varepsilon}})$ (\cref{lem:phase_refl}) with precision $\hat{\varepsilon}^{3/2}/\sqrt{\alpha W_-(\mathscr P')}$ and accuracy $\hat{\varepsilon}^2$ to $(\ket{0}\ket{\rho_x})_A\ket{0}_B$  and return the result\;
       Measure $\ket{\psi}$ in the standard basis to obtain a guess $\hat{y}$ for
       $f(x)$\;
       If $\hat{y}$ passes the additional verification, return $\hat{y}$\;

      }
      }
      Return ``error''\;

    \caption{}
    \label{alg:StateConvVerify}
\end{algorithm}{}

 We first analyze the case that $w_+(\mathscr P,x)W_-(\mathscr P)\leq w_-(\mathscr P,x)W_+(\mathscr P)=w_+(\mathscr P^\dagger,x)W_-(\mathscr P^\dagger)$.
From \cref{lem:finalStateAnalysis}, we have that when $\alpha\geq \wp{\mathscr P}{x}$, then the output state $\ket{\psi}$ of Line 7 of \cref{alg:StateConvVerify} satisfies
\begin{equation}
\|\ket{\psi}-\ket{1}\ket{f(x)}\ket{0}\|\leq 6\sqrt{\hat{\varepsilon}}=\sqrt{\delta}.
\end{equation}
This gives us 
\begin{equation}
1 -  \frac{\delta}{2} \leq Re\left(\bra{\psi}(\ket{1}\ket{f(x)}\ket{0})\right)\leq\left|\bra{\psi}(\ket{1}\ket{f(x)}\ket{0})\right|.
\end{equation}
Taking the square of both sides and using that $1-\delta \leq \left(1 -  \frac{\delta}{2}\right)^2$, we have
\begin{equation}
1 -  \delta \leq \left|\bra{\psi}(\ket{1}\ket{f(x)}\ket{0})\right|^2.
\end{equation}
Since $\ket{1}\ket{f(x)}\ket{0}$ is a standard basis state, if we measure $\ket{\psi}$ in the standard basis, this implies that probability that we measure the second register to be
$\ket{f(x)}$ is at least $1-\delta.$

Once we measure $f(x)$, we can verify it with certainty using constant
additional queries. Thus, our success probability is at least $1-\delta$ at
a single round when $\alpha\geq w_+$ (which we only reach if we haven't
already correctly evaluated $f(x)$), and so our overall probability of success must be
at least $1-\delta$ (from Line 1 of \cref{alg:StateConvVerify}). This is because further rounds (if they happen) will only increase
our probability of success.

To calculate the average query complexity, we note that the $i^\textrm{th}$ round
uses
\begin{equation}
O\left(\frac{2^{i/2}}{\hat{\varepsilon}^{3/2}}\log\left(\frac{1}{\hat{\varepsilon}}\right)+1\right)
\end{equation}
queries, where the $1$ is from the verification step, which we henceforward absorb into the big-Oh notation.

We make a worst case assumption that
the probability of measuring outcome $\ket{f(x)}$ in a round when 
$\alpha< \wp{\mathscr P}{x}$, or equivalently, at a round $i$ when $i<i^*$, is 0, so these rounds contribute
\begin{equation}
\sum_{i=0}^{i^*-1}O\left(\frac{2^{i/2}}{\hat{\varepsilon}^{3/2}}\log\left(\frac{1}{\hat{\varepsilon}}\right)\right)=O\left(\frac{2^{i^*/2}}{\hat{\varepsilon}^{3/2}}\log\left(\frac{1}{\hat{\varepsilon}}\right)\right)=O\left(\frac{2^{i^*/2}}{\delta^{3/2}}\log\left(\frac{1}{\delta}\right)\right)
\end{equation}
queries to the average query complexity, where in the last equality, we've 
replaced $\hat{\varepsilon}$ with $\delta$.

At each additional round $i$ for $i\geq i^*$, we have a $1-\delta$ probability of successfully returning $f(x)$, conditioned on reaching that round, and $\delta$ probability of continuing to the next round. This gives us an average query complexity on input $x$ of
\begin{align}
\sum_{i=i^*}^{T}&\left(1-\delta\right)\delta^{(i-i^*)}O\left(\frac{2^{i/2}}{\delta^{3/2}}\log\left(\frac{1}{\delta}\right)\right)\nonumber\\
&=O\left(\left(\frac{2^{i^*/2}}{\delta^{3/2}}\log\left(\frac{1}{\delta}\right)\sum_{i=i^*}^T\left(\sqrt{2}\delta)\right)^{i-i^*}\right)\right),
\end{align}
By our assumption that $\delta<2^{1/2}$, the summation is bounded by a constant. Thus the average query complexity on input $x$ is
\begin{equation}
\mathbb{E}[Q_x]=O\left(\frac{2^{i^*/2}}{\delta^{3/2}}\log\left(\frac{1}{\delta}\right)\right)=
O\left(\frac{\sqrt{\wp{\mathscr P}{x}W_-(\mathscr P)}}{\delta^{3/2}}\log\left(\frac{1}{\delta}\right)\right).
\end{equation}

When $\wp{\mathscr P}{x}W_-(\mathscr P)> \wm{\mathscr P}{x}W_+(\mathscr P)=\wp{\mathscr P^\dagger}{x}W_-(\mathscr P^\dagger)$, we get
the same expression except with $\mathscr P$ replaced by $\mathscr{P}^\dagger,$ which by \cref{lem:complement} gives us the claimed query complexity.
\end{proof}


\subsection{Quantum Advantages for Decision Trees with Advice}
\label{sec:tree_application}


Montanaro showed that when searching for a single marked item, if there is 
a power law distribution on the location of the item, then a quantum algorithm can achieve a (super)exponential 
speed-up in average query complexity over the best classical algorithm
\cite{montanaro2010quantum}. He called this ``searching with advice,'' as
in order to achieve the best separations between quantum and classical performance,
 the algorithm had to know an ordering of the inputs such that
 the probability of finding the marked item was non-increasing, the ``advice.''

In this section, we generalize Montanaro's result to
decision tree algorithms, and use this generalization to 
prove a superpolynomial and exponential speed-up for several additional search problems.
We use a decision tree construction similar to that of
Beigi and Taghavi \cite{beigiQuantumSpeedupBased2019}. 

A classical, deterministic query algorithm that
evaluates $f:X\rightarrow [m]$ for $X\subseteq [q]^n$ is given access to an
oracle $O_x$, for $x\in X$, and uses a single query to learn $x_i$, the
$i^\textrm{th}$ bit of $x.$
We can describe the sequence of queries this classical algorithm makes as a
directed tree $\sop T$, a decision tree,
with vertex set $V(\sop T)$ and directed edge set 
$E(\sop T)$. Each non-leaf vertex $v$ of $V(\sop T)$ is
associated with an index $J(v)\in [n]$, which is the index of $x$ that is queried
when the algorithm reaches that vertex. The algorithm follows the edge labeled by $x_{J(v)}$ (the query result) from $v$ to another vertex in $V(\sop T)$.
Each leaf is labelled by an element of $[m]$, which is the value that the
algorithm outputs if it reaches that leaf. Let $path(\sop T, x)$ be the sequence
of edges in $E(\sop T)$ that are followed on input $x$ when queries
are made starting from the root of $\sop T$. We say that $\sop T$ 
decides $f:X\rightarrow [m]$ if the leaf node on $path(\sop T, x)$ is labelled by $f(x)$, for all $x\in X$.

For a non-leaf vertex $v\in \sop T$, each edge $(v,v')\in E(\sop T)$ is labeled by a subset of $[q]$, that we denote $Q(v,v')$. 
Then if a vertex $v$ is visited in $path(\sop T,x)$, the algorithm chooses to follow the edge $(v,v')$ to vertex $v'$, if $x_{J(v)}\in Q(v,v').$ We require that $\{Q(v,v')\}$ for all edges $(v,v')$ leaving vertex $v$ form a partition of $[q]$, so that there is always exactly one edge
that the algorithm can choose to follow based on the result of the 
query at vertex $v$.

To create a quantum algorithm from such a decision tree $\sop T$, we label each edge 
$e\in E(\sop T)$ with a weight $r(e)\in \mathbb{R}^+$ and a color 
$c(e)\in\{red, black\}$, such that all edges 
coming out of a vertex $v$ with the same color have the same weight. 
There must be exactly one edge leaving each non-leaf vertex that is black,
and the rest must be red.
We denote by $r(v,black)$ the weight of the black
edge leaving $v$, and $r(v,red)$ the weight of any red edge(s) leaving $v$.
If there are no red edges leaving $v$, we set $r(v,red)=\infty.$ 
(In Ref. \cite{beigiQuantumSpeedupBased2019}, the red and black weights are the same
throughout the entire tree, instead of being allowed to depend on $v$. We note a similar flexibility in assigning weights is used in Ref. \cite{taghavi2022simplified}.) 
Using these weights we design a converting vector set to decide $f$:
\begin{restatable}{lemma}{decisionTreeWitness}\label{lem:decision_tree_complexity} 
Given a decision tree $\sop T$ that decides a function $f:X\rightarrow [m]$,
for $X\in [q]^n$, with weights $r(e)\in \mathbb{R}^+$ for each edge $e\in E(\sop T)$,
then there is a converting vector set $\mathscr P$ that on input $x\in X$ converts 
$\ket{\rho_{x}}=\ket{0}$ to $\ket{\sigma_{x}}=\ket{f(x)}$ such that
\begin{align}
\wp{\mathscr P}{x}&\leq\sum_{e\in path(\sop T, x)}4r(e),\nonumber\\
\wm{\mathscr P}{x}&\leq \sum_{(v,u)\in path(\sop T,x)}\frac{4}{r(v,red)}
+\sum_{\substack{(v,u)\in path(\sop T,x)\\
c(v,u)=red}}\frac{4}{r(v,black)}.\label{eq:treeWitnesses}
\end{align}
\end{restatable}

\begin{proof}
In this proof $\ket{u_{xj}}$ and $\ket{v_{yj}}$, with double subscripts, refer to converting
vector sets, and $u$, $v$ with single or no subscripts refer to vertices.
We use essentially the same
construction as in Ref. \cite{beigiQuantumSpeedupBased2019}, but with a slightly
different analysis because
of our generalization to weights that can change throughout the tree.

We will make use of the unit vectors 
$\{\ket{\tilde{\mu}_{i,d}}\}_{i\in [d]}$ and 
$\{\ket{\tilde{\nu}_{i,d}}\}_{i\in [d]}$, defined in 
\cite{beigiQuantumSpeedupBased2019}, which are scaled versions of the vectors in 
\cref{eq:mu_nu}, which have the properties that 
$\forall i\in [d],\left\|\ket{\tilde{\mu}_{i,d}}\right\|^2,\left\|\ket{\tilde{\nu}_{i,d}}\right\|^2\leq 2$ and $\braket{\tilde{\mu}_{i,d}}{\tilde{\nu}_{j,d}}=1-\delta_{i,j}$.

First note that we can assume that on
any input $x\in X$, for any index $j$, there is at most a single vertex on in $path(\sop T,x)$ at which $j$ is queried. Otherwise
the tree would query the same index twice, which would be a non-optimal tree.
 Then we define a converting vector set on 
 $\mathbb{C}^{|V(\sop T)|}\otimes \mathbb{C}^{2}\otimes\mathbb{C}^{|V(\sop T)|}
 \otimes\mathbb{C}^{m}$ as follows
 \begin{align}
 \label{eq:pos_wt_witness_DT}
 \ket{v_{xj}}&=
 \begin{cases}
 \sqrt{r(v,u)}\ket{v}\ket{c(v,u)}\ket{\tilde{\mu}_{u,|V(\sop T)|}}\ket{\tilde{\mu}_{f(x),m}} &\textrm{ if } (v,u)\in path(\sop T,x), \textrm{ and } J(v)=j\\
 0&\textrm{otherwise}
 \end{cases}
 \end{align}
 and
 \begin{align}
  \label{eq:neg_wt_witness_DT}
 \ket{u_{xj}}&=
 \begin{cases}
 \frac{1}{\sqrt{r(v,red)}}\ket{v}\ket{red}\ket{\tilde{\nu}_{u,|V(\sop T)|}}\ket{\tilde{\nu}_{f(x)}} &\textrm{ if } (v,u)\in path(\sop T,x), J(v)=j,c(v,u)=black\\
 \ket{v}\left(\frac{\ket{red}}{\sqrt{r(v,red)}}+\frac{\ket{black}}{\sqrt{r(v,black)}}\right)\ket{\tilde{\nu}_{u,|V(\sop T)|}}\ket{\tilde{\nu}_{f(x)}} &\textrm{ if } (v,u)\in path(\sop T,x), J(v)=j,c(v,u)=red\\
 0&\textrm{otherwise}
 \end{cases}
 \end{align}

From \cref{def:converting_vector_set}, \cref{eq:filteredNorm}, we want that 
\begin{align}
\forall x,y\in X,\quad (\rho-\sigma)_{xy}=\sum_{j\in[n]:x_j\neq y_j}\braket{u_{xj}}{v_{yj}}.
\end{align}
For evaluating a discrete function $f$, we have $\ket{\rho_x}=\ket{0}$, and 
$\ket{\sigma_x}=\ket{f(x)}$, so $(\rho-\sigma)_{xy}=1-\delta_{f(x),f(y)}.$

Now if $f(x)=f(y)$, then because $\braket{\tilde{\mu}_{f(x),m}}{\tilde{\nu}_{f(y),m}}=0$,
we have $\forall j\in [n]$ $\braket{u_{xj}}{v_{yj}}=0$, so 
\\$\sum_{j\in [n]:x_j\neq y_j}\braket{u_{xj}}{v_{yj}}=0$
, as desired.

If $f(x)\neq f(y)$, then there must be some vertex in the decision tree at which $path(\sop T,x)$
and $path(\sop T,y)$ diverge. Let's call this vertex $v^*$, and assume that $J(v^*)$, the
index of the input queried at vertex $v^*$, is $j^*$. Let $(v^*,u_1)$ be the edge on 
$path(\sop T,x)$ and $(v^*,u_2)$ be the edge on 
$path(\sop T,y)$. This means that $x_{j^*}\in Q(v^*,u_1)$, and
$y_{j^*}\in Q(v^*,u_2)$ and since $Q(v^*,u_1)$ and $Q(v^*,u_2)$ are part of a partition,
 this implies $x_{j^*}\neq y_{j^*}.$
Then if $c(v^*,u_1)=black$, we must have $c(v^*,u_2)=red$, while if $c(v^*,u_1)=red$, we can have $c(v^*,u_2)\in \{black,red\}$. In either case, we see from \cref{eq:pos_wt_witness_DT,eq:neg_wt_witness_DT} that
\begin{align}
\braket{u_{xj^*}}{v_{yj^*}}=1.
\end{align}

Now for all other $j\in [n]$ with $j\neq j^*$, we have $\braket{u_{xj}}{v_{yj}}=0$, which
we can prove by looking at the following cases:
\begin{itemize}
\item There is no vertex in $\sop T$ where
$j$ is queried for $x$ or $y$, which results in $\ket{v_{xj}}=0$ or $\ket{u_{xj}}=0$, and
so $\braket{u_{xj}}{v_{yj}}=0$.
\item The index $j$ is queried for both $x$ and $y$ before their paths in $\sop T$ diverge, at a vertex
$v$ where the paths for both $x$ and $y$ then travel to a vertex $u$,
in which case, $\braket{u_{xj}}{v_{yj}}=0$, since $\braket{\tilde{\mu}_{u,|V(\sop T)|}}{\tilde{\nu}_{u,|V(\sop T)|}}=0$.
\item The index $j$ is queried for both $x$ and $y$ after their paths in $\sop T$ diverge, at
a vertex $v_1$ for $x$ and a vertex $v_2$ for $y$, in which case $\braket{u_{xj}}{v_{yj}}=0$,
since $\braket{v_1}{v_2}=0$.
\end{itemize}

Putting this all together for $j=j^*$ and $j\neq j^*$, we have 
\begin{align}
\sum_{j\in [n]:x_j\neq y_j}\braket{u_{xj}}{v_{yj}}=1.
\end{align}

Now to calculate the positive and negative witness sizes. For the positive witness
size, we have
\begin{align}
\wp{\mathscr P}{x}&=\sum_{j\in[n]}\|\ket{v_{xj}}\|^2\nonumber\\
&=\sum_{(v,u)\in path(\sop T, x)}\left\|\sqrt{r(v,u)}\ket{v}\ket{c(v,u)}\ket{\tilde{\mu}_{u,|V(\sop T)|}}\ket{\tilde{\mu}_{f(x),m}}\right\|^2\nonumber\\
&\leq \sum_{e\in path(\sop T, x)}4 r(e),
\end{align}
where in the second line, 
we have used \cref{eq:pos_wt_witness_DT} and that $\sop T$ will query each index at most once,
according to the vertices that are encountered in $path(\sop T,x)$,
and in the final line, we have used that $\|\ket{\tilde{\mu}_{u,|V(\sop T)|}}\|^2,\|\ket{\tilde{\mu}_{f(x),m}}\|^2\leq 2.$

For the negative witness size, note that again for input $x$,
$\sop T$ will query each index of $x$ at most once, according to the 
vertices that are encountered in $path(\sop T)$. Then if index $j$ of $x$ is queried 
at vertex $v$, where $(v,u)\in path(\sop T,x)$, and $c(v,u)=black$ then 
\begin{align}
\|\ket{u_{xj}}\|^2\leq 4\frac{1}{r(v,red)}
\end{align}
while if $c(v,u)=red$, then
\begin{align}
\|\ket{u_{xj}}\|^2\leq 4\left(\frac{1}{r(v,red)}+\frac{1}{r(v,black)}\right).
\end{align}
Thus
\begin{align}
\wm{\mathscr P}{x}&=\sum_{j\in[n]}\|\ket{u_{xj}}\|^2\nonumber\\
&=\sum_{\substack{(v,u)\in path(\sop T,x)\\
c(v,u)=black}}\frac{4}{r(v,red)}
+\sum_{\substack{(v,u)\in path(\sop T,x)\\
c(v,u)=red}}\left(\frac{4}{r(v,red)}+\frac{4}{r(v,black)}\right)\nonumber\\
&=\sum_{(v,u)\in path(\sop T,x)}\frac{4}{r(v,red)}
+\sum_{\substack{(v,u)\in path(\sop T,x)\\
c(v,u)=red}}\frac{4}{r(v,black)}.
\end{align}
\end{proof}

In \cref{thm:decision_tree}, we use \cref{lem:decision_tree_complexity} 
to derive average quantum and classical query 
separations based on classical decision trees. 

\begin{restatable}{theorem}{decisionTree}\label{thm:decision_tree} If $\sop T$ is a decision tree that
 decides $f:X\rightarrow[m]$ for $X\subseteq [q]^n$, with optimal average
 classical query complexity for the distribution $\{p_x\}_{x\in X}$, and $\sop T$ has
 a coloring such that there are at most $G$ red edges on
 any path from the root to a leaf, then the average quantum query complexity
 of deciding $f(x)$ with bounded error is 
 \begin{equation}
 O\left(\sum_{x\in X}p_x\sqrt{G|path(\sop T,x)|\log^3(n)}\right).
 \end{equation}
If it is possible to verify a potential output $\hat{y}$ as correctly being $f(x)$ using constant queries, then 
the average quantum query complexity
 of deciding $f(x)$ with bounded error is 
 \begin{equation}
 O\left(\sum_{x\in X}p_x\left(\sqrt{G|path(\sop T,x)|\log(n)}\right)\right).
 \end{equation}
The average classical query complexity of deciding $f(x)$ with bounded
error is
\begin{equation}\label{eq:class_decision}
\sum_{x\in X}p_x|path(\sop T,x)|.
\end{equation}
\end{restatable}

\begin{proof} 
The average classical query complexity comes from the fact that on input $x$,
which occurs with probability $p_x$, the algorithm uses $|path(\sop T,x)|$
queries, since each edge on the path of the decision tree corresponds
to a single additional query. By assumption, $\sop T$ is optimal for
the distribution $\{p_x\}$, giving the complexity as in \cref{eq:class_decision}.

For the quantum algorithm, we will assign weights to
each edge in $\sop T$, and then use \cref{lem:decision_tree_complexity} to create and analyze
a state conversion algorithm. Then we will then apply \cref{thm:stateConvNew} and
\cref{thm:verify_state_conversion} to achieve better complexity on easier inputs. 

For each black edge $e$ in $\sop T$ let $r(e)=G$. For each red 
edge $e$, let $r(e)=l(e)$, where $l(e)$ is the number of edges on the path in $\sop T$ from 
the root to $e$, including $e$. Let $\mathscr P$ be the 
converting vector set from \cref{lem:decision_tree_complexity} that converts 
$\ket{\rho_x}=\ket{0}$ to $\ket{\sigma_x}=\ket{f(x)}$, based on $\sop T$. 

We first analyze $\wp{\mathscr P}{x}$. By \cref{lem:decision_tree_complexity},
\begin{align}
\wp{\mathscr P}{x}
=&\sum_{\textrm{black } e\in path(\sop T,x)}4r(e)+\sum_{\textrm{red }e\in path(\sop T,x)}4r(e)\nonumber\\
=&\sum_{\textrm{black } e\in path(\sop T,x)}4G+\sum_{\textrm{red }e\in path(\sop T,x)}4l(e)\nonumber\\
\leq&4G|path(\sop T,x)|+\sum_{\textrm{red }e\in path(\sop T,x)}4|path(\sop T,x)|\nonumber\\
= &O\left(G|path(\sop T,x)|\right),
\end{align} 
where in the second-to-last line, we've used that the total number
of black or red edges in the path is $|path(\sop T,x)|$. In the last line,
we've used that the number of red edges on any path is at most $G$. 
This implies that $W_+(\mathscr P)=O(nG).$

Now to analyze $\wm{\mathscr P}{x}$. From \cref{lem:decision_tree_complexity},
\begin{align}
\wm{\mathscr P}{x}&\leq \sum_{(v,u)\in path(\sop T,x)}\frac{4}{r(v,red)}
+\sum_{\substack{(v,u)\in path(\sop T,x)\\
c(v,u)=red}}\frac{4}{r(v,black)}
\nonumber\\
&= \sum_{(v,u)\in path(\sop T,x)}\frac{4}{l(v,u)}
+\sum_{\substack{(v,u)\in path(\sop T,x)\\
c(v,u)=red}}\frac{4}{G}
\nonumber\\
&\leq \sum_{i=1}^{|path(\sop T,x)|}\frac{4}{i}
+4
\nonumber\\
&=O(\log(n)).
\end{align}

Now applying \cref{thm:stateConvNew} with $\varepsilon,\delta=\Theta(1)$ and $W_+(\mathscr P)=O(n)$ gives us a bounded
error algorithm with an average
query complexity of $O\left(\sqrt{G|path(\sop T,x)|\log^3(n)}\right)$ on input $x$. 
On average over $x\in X$, we obtain an average query complexity of 
$ O\left(\sum_{x\in X}p_x\sqrt{G|path(\sop T,x)|\log^3(n)}\right).$

When there is a way to verify $f(x)$ using a constant queries,
we can apply \cref{thm:verify_state_conversion} with $\delta=\Theta(1)$ 
to give us a bounded error algorithm with an average
query complexity of $O\left(\sqrt{G|path(\sop T,x)|\log(n)}\right)$
on input $x$. 
On average over $x\in X$, we obtain an average query complexity of 
$ O\left(\sum_{x\in X}p_x\left(\sqrt{G|path(\sop T,x)|\log(n)}\right)\right).$
\end{proof}

We now use \cref{thm:decision_tree} to show an average quantum advantage for
two problems related to searching: searching for two marked items in a list and searching for the first marked item in a list with two marked items:

\begin{restatable}{theorem}{searchSpeedup}\label{thm:search_speedup} For the problem of finding $2$
bits with value $1$ in an  $n$-bit string, there is a distribution for which there is an exponential advantage in average quantum  query complexity over average
 classical query complexity. For the problem of finding the first $1$-valued bit
 in an $n$-bit string with at most two $1$-valued bits, there is a distribution for which there is
 a superpolynomial advantage in average quantum  query complexity over average classical
 query complexity.
\end{restatable} 

The proof uses a decision tree $\sop T$ that checks the
$n$ bits of the string in order until either one or two $1$-valued bits are found. The tree
for finding two $1$-valued bits is shown in \cref{fig:tree}. Each time a $1$-valued bit is found, the edge
that the algorithm traverses is colored red. Then $G(\sop T)=1$ or $2$, depending on the problem, so \cref{thm:decision_tree} tells us the average query complexity will be small when the two marked bits occur early in the list, resulting in a short path for that input. We combine
this idea with distributions that are modified versions of
a power law distribution;
this power law distribution is tailored to allow a quantum algorithm, which has at most a
quadratic advantage on any particular input, but which only uses constant queries
on the easiest inputs, to 
achieve an exponential/superpolynomial advantage on average \cite{montanaro2010quantum}.

To obtain these results, we explicitly analyze a particular distribution of bit strings with Hamming weight $1$ or $2$, but we expect that similar techniques could be applied to additional distributions that include strings with larger Hamming weights. The bottleneck is not analyzing the quantum algorithm, but in proving properties of the optimal classical algorithm.

\begin{figure}[h]
\centering
\includegraphics[scale=1.2]{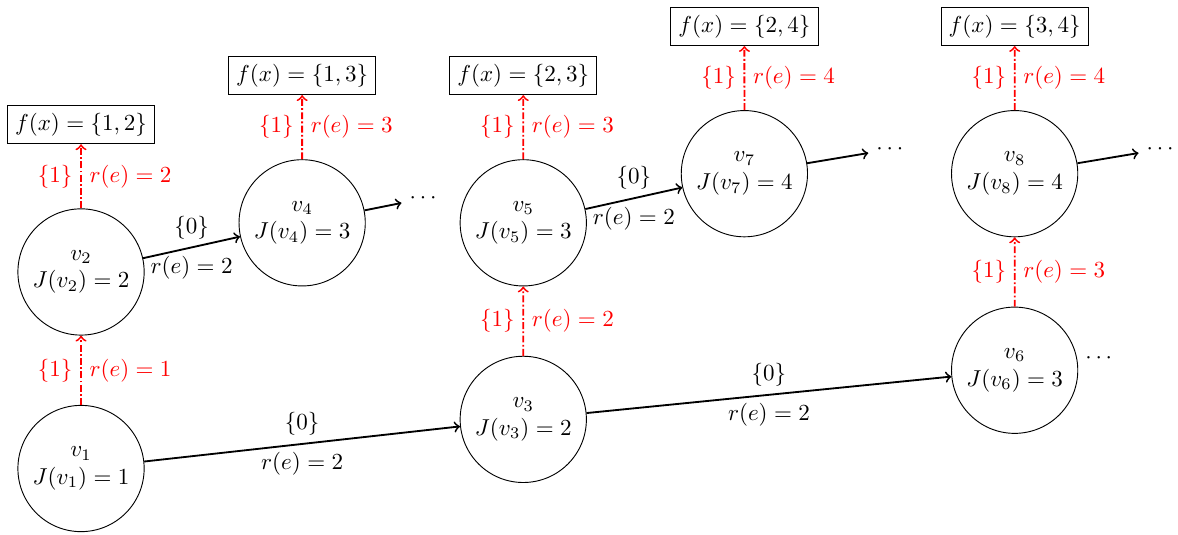}
\caption{The decision tree we use to design the quantum algorithm for
 finding two bits with value $1$. Each vertex is labelled by its name ($v_i$) for some $i$, and $J
 (v_i)$, which is the bit of the input that is queried if the algorithm
 reaches that vertex of the tree. Each edge $(v_i,v_j)$ is labelled by $Q
 (v_i,v_j)$, which is the set in curly brackets alongside each edge. The algorithm
 follows the edge $(v_i,v_j)$  from vertex $v_i$ if the value of the query
 made at vertex $v_i$ is contained in $Q(v_i,v_j)$. Each edge is also labelled
 by its weight, $r(e)$, and is also colored red or black (and red edges are
 additionally rendered with dot-dashes.) Black edges all have weight $G
 (\sop T)$, which in this case is $2$. Each red edge has a weight that is
 equal to the number of edges on the path from the root $v_1$ to that edge,
 inclusive. The vertex $v_1$ is the root, and each leaf (denoted as a rectangular vertex) is labelled by the 
 output of the algorithm on that input.}
\label{fig:tree}
\end{figure}

\begin{proof}[Proof of \cref{thm:search_speedup}]
For both problems (finding two marked items, or finding the first marked item), we consider the following distribution of $n$-bit strings. For the purposes of describing the distribution, for each string we label an index among $\{2,\dots, n\}$ as the \textit{dividing index}, although the position of this index is not known to us when we actually sample a string. Let $E_i^*$ be the event that $i$ is the dividing index of the sampled string. Then $p(E_i^*)\propto (i-1)^k$ for a constant $k$ such that $-3/2>k>-2$. For a string with dividing index $i$, 
all bits with index greater than $i$ have value $0$, and exactly one bit with index less than $i$ has value $1$, chosen uniformly at random. In the case that we are trying to find
two bits with value $1$, bit at the dividing index also has value $1$. In the case that we are trying to find the first bit with value $1$, we set the bit at the dividing index
to have value $1$ with probability $p_+$. In \cref{app:explicit} \cref{lem:class_in_order}, we prove that
given these distributions and problems, querying the bits of the string in order solves these problems with query complexity within $1$ query of the optimal classical strategy, and so analyzing the query complexity of an
algorithm that queries the bits in order is sufficient for bounding the optimal asymptotic classical query complexity.

We first analyze the case of finding both $1$-valued bits. That is, on input $x$, we want to evaluate the function $f(x)=\{i_1,i_2\}$, where $i_1$ and $i_2$ are the indices of the two bits of $x$ that have value $1$.

Consider a decision tree that queries the bits of the string in increasing order. This corresponds to a tree
$\sop T$ as shown in \cref{fig:tree}, where $|path(\sop T,x)|=i$, where $i$ is the index of the $2^
{nd}$ $1$-valued bit in the string. 
We label an edge $(v,v')$ of this tree as red whenever $1\in Q(v,v')$; that is, edges that are traversed when a $1$-valued bit is found are
colored red. Then $G(\sop T)=2.$ 

For this problem, one can verify whether an output of the quantum algorithm is correct using an additional $2$ queries; query the two output indices to ensure there is a $1$ at each position, in which case, one knows with certainty that the output is correct. Otherwise, the output is incorrect.

Thus by \cref{thm:decision_tree}, since $p(E_i^*)$ is the probability of finding the $2^\textrm{nd}$ (final) $1$-valued bit at index $i$,
the average quantum query complexity is
\begin{equation}
O\left(\sum_{i=2}^{n}p(E_i^*)\left(\sqrt{i2\log(n)}\right)\right). \label{eq:quantum_query_not_analyzed3}
\end{equation}
In \cref{app:explicit} \cref{lem:sum_calculations}, we show that
\begin{equation}
O\left(\sum_{i=2}^{n}p(E_i^*)\left(\sqrt{i}\right)\right)=O(1),
\end{equation}
so the quantum query complexity of finding both $1$-valued inputs is 
\begin{equation}
O(\log^{1/2}(n)).
\end{equation}



Since the optimal classical strategy (to within one query) is to query the bits of the string in order (see \cref{app:explicit} \cref{lem:class_in_order}), and such an algorithm will terminate after $i$ queries if $i$ is the dividing index, the asymptotic classical query complexity is $\Omega\left(\sum_{i=r}^{n}p(E_i^*)i\right)$, and in \cref{lem:sum_calculations}, we prove
\begin{align}
\sum_{i=r}^{n}p(E_i^*) i&=\Omega \left(n^{k+2}\right).
\end{align}
Thus we have an exponential quantum improvement, as the average quantum query complexity is
$O\left(\log^{1/2}(n)\right)$ compared to the classical $\Omega(n^{k+2})$,
 for $-2<k< -3/2$.

Now we consider the problem of finding the first $1$-valued bit in the input string for the distribution described at the beginning of this lemma.
Consider a decision tree that queries the bits of the string in increasing order. This corresponds to a tree 
$\sop T$ where $|path(\sop T,x)|=i$, where $i$ is the index of the first $1$-valued bit in the string. We label an edge $(v,v')$ of this tree as red whenever $1\in Q(v,v')$; that is, edges that are traversed when a $1$-valued bit is found are
colored red. Then $G(\sop T)=1.$ However, now when this algorithm returns an index, we can no longer easily verify that it is the \textit{first}
$1$-valued bit. From \cref{thm:decision_tree}, if $E_i^\dagger$ is the event that the index of the first $1$-valued bit is $i$, we have that the quantum
query complexity is
\begin{equation}\label{eq:clas_qu_first}
O\left(\sum_{i\in n}p(E_i^\dagger)\left(\sqrt{i\log(n)}\right)\right).
\end{equation}

On the other hand, because the optimal classical strategy (to within one query) is to query the bits
of the input string in order, the average classical query complexity is
$\Omega\left(\sum_{i\in n}p(E_i^\dagger)i\right),$ since the algorithm will terminate after $i$ queries if the first $1$-valued bit is at index $i$.
In \cref{lem:sum_calculations}, we show 
\begin{align}
&\sum_{i\in n}p(E_i^\dagger)i=\Omega(n^{k+2}),\nonumber\\
&\sum_{i\in n}p(E_i^\dagger)\sqrt{i}=O(1),
\end{align}
so the average classical query complexity is $\Omega(n^{k+2})$ and the average quantum
query complexity is $O\left(\polylog(n)\right),$ a superpolynomial improvement. 
\end{proof}


\paragraph{Acknowledgments}
We thank Stacey Jeffery for valuable discussions, especially for her preliminary notes on span program negation, and several past referees for insightful suggestions.

\paragraph{Funding}
All authors were supported by the U.S. Army Research Office 
under Grant Number W911NF-20-1-0327. The views and conclusions contained in
this document are those of the authors and should not be interpreted as
representing the official policies, either expressed or implied, of the Army
Research Office or the U.S. Government. The U.S. Government is authorized to
reproduce and distribute reprints for Government purposes notwithstanding any
copyright notation herein.

\bibliography{spanPrograms.bib}
\bibliographystyle{alphaurl}


\appendix
\section{Complements of Span Programs and Converting Vector Sets}\label{app:sec2}

We first prove that, given a span program that decides $f:X\rightarrow\{0,1\}$,
we can create a span program that decides $f^{\neg }$ while preserving witnesses
sizes for each $x$.

\SPduals*

\begin{proof}
We first define $H'$, starting from $H'_{j,a}$:
\begin{equation}
H'_{j,a}=\textrm{span}\{\ket{v}: \ket{v}\in H_j \textrm{ and } \ket{v}\in H_{j,a}^\perp\},
\end{equation}
where $H_{j,a}^\perp$ is the orthogonal complement of $H_{j,a}.$ We define $H'_j=\sum_{a\in [q]}H'_{j,a}$, and $H'_{true}=H_{false}$ and $H'_{false}=H_{true}$.
Then
\begin{equation}
H'=H'_1\oplus H'_2\cdots H'_n\oplus H'_{true}\oplus H'_{false}.
\end{equation}
Let $\ket{\tilde{0}}$ be a vector that is orthogonal to $H$ and $V$, and
define $V'=H\oplus\textrm{span}\{\ket{\tilde{0}}\}$ and $\tau'=\ket{\tilde{0}}.$
Finally, set
\begin{equation}
A'=\ketbra{\tilde{0}}{w_0}+\Lambda_A,
\end{equation}
where $\Lambda_A$ is the projection onto the kernel of $A$, $\Pi_H$ is the projection onto $H$, and 
\begin{equation}\label{eq:complement_min_wit}
\ket{w_0}=\argmin_{\substack{\ket{v}\in H: A\ket{v}=\tau}}\|\ket{v}\|.
\end{equation}

Let $x\in X$ be an input with $f(x)=1$, so $x$ has a positive witness
$\ket{w}$ in $\sop P$. We will show $\omega'=\bra{\tilde{0}}+\bra{w}$ is a negative
witness for $x$ in $\sop P^\dagger$. Note $\omega'\tau'=1$, and also,
\begin{align}
\omega' A'&=(\bra{\tilde{0}}+\bra{w})(\ketbra{\tilde{0}}{w_0}+\Lambda_A)\\
&=\bra{w_0}+\bra{w}\Lambda_A\\
&=\bra{w},
\end{align}
where in the second line, we have used that $\bra{\tilde{0}}\Lambda_A=0$ since
$\Lambda_A$ projects onto the $H$ subspace, 
and $\braket{w}{\tilde{0}}=0$ because $\ket{\tilde{0}}$ is orthogonal to $H$.
The final line follows from \cite[Definition
2.12]{itoApproximateSpanPrograms2019}, which showed that every positive
witness can be written as $\ket{w}=\ket{w_0}+\ket{w^\perp}$, where
$\ket{w^\perp}$ is in the kernel of $A$ and $\ket{w_0}$ is orthogonal to the
kernel of $A$. 

Then $\bra{w}\Pi_{H'(x)}=0$, because $\ket{w}\in H(x)$, and
$H'(x)$ is orthogonal to $H(x),$ so $\omega'$ is a negative witness for $x$ in
$\sop P^\dagger$.  Also, $\|\omega' A'\|^2=\|\ket{w}\|^2$, so the witness size of
this negative witness in $\sop P^\dagger$ is the same as the corresponding
positive witness in $\sop P.$ This implies $\wm{\sop P^\dagger}{x}\leq \wp{\sop P}{x}.$

Next, we show that $\wm{\sop P^\dagger}{x}\geq \wp{\sop P}{x}$. Let $\omega'$ be 
any negative witness for $x$ in $\sop P^\dagger$. Then since $\omega'\in \sop L(H\oplus\textrm{span}\{\ket{\tilde{0}}\},\mathbb{R})$, and we must have $\omega'\tau'=\omega'\ket{\tilde{0}}=1$,
$\omega'$ must have the form $\bra{\tilde{0}}+\bra{w}$, where $\ket{w}\in H$. Then
\begin{align}
\omega' A'&=(\bra{\tilde{0}}+\bra{w})(\ketbra{\tilde{0}}{w_0}+\Lambda_A)\\
&=\bra{w_0}+\bra{w}\Lambda_A\\
&=\bra{w_0}+\bra{h}
\end{align}
where $\ket{h}$ is in the kernel of $A$. By the definition of $\ket{w_0}$, 
this implies that $A(\ket{w_0}+\ket{h})=\tau$. 
Next, for $\omega'$ to be a valid negative witness 
for $x$ in $\sop P^\dagger$, we must have that $\omega' A'\Pi_{H'(x)}=0$, which implies
that $(\omega' A')^\dagger\in H(x)$, and so $\ket{w_0}+\ket{h}\in H(x)$. Thus
$(\omega' A')^\dagger$ is a positive witness for $x$ in $\sop P$, and so $\|\omega'A'\|^2\geq \wp{\sop P}{x}$, implying 
$\wm{\sop P^\dagger}{x}\geq \wp{\sop P}{x}$.

If $f(x)=0$, there is a negative witness $\omega$ for $x$ in $\sop P$. Consider
$\ket{w'}=(\omega A)^\dagger$. Then
\begin{align}
A'\ket{w'}&=(\ketbra{\tilde{0}}{w_0}+\Lambda_A)(\omega A)^\dagger\\
&=\ket{\tilde{0}}(\omega A\ket{w_0})^\dagger\\
&=\ket{\tilde{0}}(\omega\tau)^\dagger\\
&=\ket{\tilde{0}},
\end{align}
where in the second line, we have used that $(\omega A)^\dagger$ is orthogonal to
the kernel of $A.$ Also, $\Pi_{H(x)}(\omega A)^\dagger=0$, so $\ket{w'}\in
H'(x)$. This means $\ket{w'}$ is a positive witness for $x$ in $\sop P^\dagger.$
Also, $\|\ket{w'}\|^2=\|\omega A\|^2$, so the witness size of this positive
witness in $\sop P^\dagger$ is the same as the corresponding negative
witness in $\sop P.$ This implies $\wp{\sop P^\dagger}{x}\leq \wm{\sop P}{x}.$

Now to show $\wp{\sop P^\dagger}{x}\geq \wm{\sop P}{x}$. Let $\ket{w'}$ be a positive
witness for $x$ in $\sop P^\dagger.$ Then since $\tau'=A'\ket{w'}$, we have
\begin{align}\label{eq:complement1}
\ket{\tilde{0}}&=(\ketbra{\tilde{0}}{w_0}+\Lambda_A)\ket{w'}\\
&=\ket{\tilde{0}}\braket{w_0}{w'}+\Lambda_A\ket{w'}
\end{align}
Since $\Lambda_A\ket{\tilde{0}}=0$, we must have $\Lambda_A\ket{w'}=0$,
which implies $\ket{w'}=(\omega A)^\dagger$ for some $\omega\in V$. Plugging
this into \cref{eq:complement1}, and using \cref{eq:complement_min_wit}, we have
\begin{align}
\ket{\tilde{0}}&=\ket{\tilde{0}}(\omega A\ket{w_0})^\dagger\\
&=\ket{\tilde{0}}(\omega \tau)^\dagger.
\end{align}
Thus we have $\omega\tau=1$. Since $\ket{w'}$ is a positive witness 
for $\sop P^\dagger$, $\Pi_{H'(x)}\ket{w'}=\ket{w'}$, which
implies that $\Pi_{H(x)}(\omega A)^\dagger=0,$ or equivalently, 
$\omega A \Pi_{H(x)}=0$. Thus we see that $\omega$ must in fact be a 
negative witness for $x$ in $\sop P$. Therefore, $\wp{\sop P^\dagger}{x}\geq \wm{\sop P}{x}$.

Finally, since if
$x$ has a positive witness in $\sop P$, we have shown it has a negative witness 
in $\sop P^\dagger$, and if
$x$ has a negative witness in $\sop P$, we have shown it has a negative witness in $\sop P^\dagger$. Thus $\sop P^\dagger$ does indeed decide $f^{\neg }.$
\end{proof}

The next two proofs deal with manipulating converting vector sets.
\scaleSC*
\begin{proof}
Let $\mathscr P=(\ket{u_{xj}},\ket{v_{xj}})_{x\in X,j\in[n]}$.
We scale the vectors in $\mathscr P$ to create the converting
vector set $\mathscr P'$ with $\ket{v_{xj}'}=\ket{v_{xj}}\sqrt{W_-/W_+}$ and 
$\ket{u_{xj}'}=\ket{u_{xj}}\sqrt{W_+/W_-}$. The converting vector set
$\mathscr P'$  still satisfies the
constraints of \cref{eq:filteredNorm}, but now has maximum witness sizes
$W_+'=W_-'=\sqrt{W_+W_-}$, so applying \cref{thm:StateConv} gives the result.
\end{proof}

Next we prove that given a converting vector set, we can design another converting 
vectors set that converts between the same states, but with positive and negative witness sizes 
exchanged.

\complment*
\begin{proof}
Let $\mathscr P=(\ket{v_{xj}},\ket{u_{xj}})_{x\in X,j\in[n]}$.
For all $x\in X$ and $j\in[n]$,
define
\begin{equation}
\ket{u^C_{xj}}=\ket{v_{xj}}, 
\qquad \textrm{and} \qquad \ket{v^C_{xj}}=\ket{u_{xj}}.
\end{equation}
Note $\braket{u^C_{xj}}{v^C_{yj}}=(\braket{u_{yj}}{v_{xj}})^*$. Since $\mathscr P$ satisfies the constraints of \cref{eq:filteredNorm},
\begin{align}
\sum_{j\in[n]:x_j\neq y_j}\braket{u^C_{xj}}{v^C_{yj}}
&=\sum_{j\in[n]:x_j\neq y_j}(\braket{u_{yj}}{v_{xj}})^*\nonumber\\
&=\left(\sum_{j\in[n]:x_j\neq y_j}(\braket{u_{yj}}{v_{xj}})\right)^*\nonumber\\
&=(\rho-\sigma)_{yx}^*\nonumber\\
&=(\rho-\sigma)_{xy},\nonumber\\
\end{align}
where we have used the fact that $\rho$ and $\sigma$ are Hermitian.

Thus the vectors $(\ket{v_{xj}^C},\ket{u_{xj}^C})_{x\in X,j\in[n]}$ satisfy the same constraints of \cref{eq:filteredNorm}, and thus produce the same have value in
\cref{eq:filteredNormMa}. However, now $w_+(\mathscr P,x)=w_-(\mathscr P^\dagger,x)$, and
$w_-(\mathscr P,x)=w_+(\mathscr P^\dagger,x)$.
\end{proof}

\section{Proofs for the Function Decision Algorithm}\label{app:func_proof}

\phaseEstEarly*

\begin{proof}
The proof is similar to Belovs and Reichardt \cite[Section
5.2]{belovsSpanProgramsQuantum2012} and Cade et al. \cite[Section
C.2]{cadeTimeSpaceEfficient2018} and the dual adversary algorithm of Reichardt
\cite[Algorithm 1]{reichardtReflectionsQuantumQuery2011}
~\\
\textit{Part 1:} Since $f(x)=1$, there is an positive optimal witness $\ket{w}\in H(x)$
for $x$. Then set $\ket{u}\in \tilde{H}(x)$ to be
$\ket{u}=\alpha\ket{\hat{0}}-\ket{w}.$ Then $\Pi_x\ket{u}=\ket{u}$, but also,
$\ket{u}$ is in the kernel of $\Aa{\alpha}$, because
$\Aa{\alpha}\ket{u}=\ket{\tau}-\ket{\tau}=0$.
Thus $\Lal{\alpha}\ket{u}=\ket{u}$, and so $\U{\sop P}{x}{
\alpha}\ket{u}=\ket{u}$; $\ket{u}$ is a 1-valued eigenvector of 
$\U{\sop P}{x}{\alpha}$.

We perform Phase Checking on the state
$\ket{\hat{0}}$, so the probability of measuring the state $\ket{0}_B$ in the phase register is \textit{at least} (by \cref{lem:phase_det}), the
overlap of $\ket{\hat{0}}$ and (normalized) $\ket{u}$. This is
\begin{equation}
\frac{|\braket{\hat{0}}{u}|^{2}}{\| \ket{u} \|^2} = \frac{\alpha^2}{\alpha^2+\|\ket{w}\|^2}=\frac{1}{1+\frac{\wp{\sop P}{x}}{\alpha^2}}
\end{equation}
Using our assumption that $\wp{\sop P}{x}\leq \alpha^2/C$, and a Taylor
series expansion for $C\geq 2$, the probability that we measure the state $\ket{0}_B$ in the phase register is at least $1-1/C$.

~\\ \textit{Part 2:} Since $f(x)=0$, there is an optimal negative witness
$\omega$ for $x$, and we set $\ket{v}\in \tilde{H}$ to be $\ket{v}=\alpha 
(\omega\Aa{\alpha})^\dagger$. By \cref{def:negWit}, $\omega\tau=1$, so $\ket{v}=(\bra{\hat{0}}+\alpha\omega A)^\dagger.$ Again, from \cref{def:negWit},
$\omega A \PHx=0$, so we have $\Pi_x\ket{v}=\ket{\hat{0}}.$

Then when we perform Phase Checking of the unitary $\U{\sop P}{x}{\alpha}$ to some
precision $\Theta$ with error $\epsilon$ on state $\ket{\hat{0}}$, by
\cref{lem:phase_det}, we will measure $\ket{0}_B$ in the phase register with probability at most
\begin{equation}\label{eq:Pbound}
\left\|P_{\Theta}(\U{\sop P}{x}{\alpha})\ket{\hat{0}}\right\|^2+\epsilon=\left\|P_{\Theta}(\U{\sop P}{x}{\alpha})\PtHx\ket{v}\right\|^2+\epsilon.
\end{equation}

Now $\ket{v}$ is orthogonal to the kernel of $\Aa{\alpha}$. (To see this, note that if $\ket{k}$ is in the kernel of $\Aa{\alpha}$, then $\braket{v}{k}=\alpha\omega \Aa{\alpha}\ket{k}=0$.)
Applying \cref{spec_gap_lemm}, and setting
$\Theta=\sqrt{\frac{\epsilon}{\alpha^2 W_-}}$, we have
\begin{equation}\label{eq:negCase1}
\left\|P_{\Theta}(\U{\sop P}{x}{\alpha})\PtHx\ket{v}\right\|^2\leq\frac{\epsilon}{4\alpha^2 W_-}\left\|\ket{v}\right\|^2.
\end{equation}

To bound $\left\|\ket{v}\right\|^2$, we observe that 
\begin{equation}\label{eq:vbound}
\left\|\ket{v}\right\|^2=\left\|\bra{\hat{0}}+\alpha\omega A\right\|^2= 1+\alpha^2\wm{\sop P}{x}\leq 
1+\alpha^2W_-\leq2\alpha^2W_-,
\end{equation}
where we have used our assumption that $\alpha^2W_-\geq 1$.

Plugging \cref{eq:negCase1,eq:vbound} into \cref{eq:Pbound}, we find that the
probability of measuring $\ket{0}_B$ in the phase register  is at most
$\frac{3}{2}\epsilon$, as claimed.
\end{proof}

\section{Analyzing the Distributions for Search Problems}
\label{app:explicit}

In the lemmas in this section, we analyze the following distribution on $n$-bit strings with Hamming weight $2$ or $1$. Each string has a
``dividing index'' which is a bit position between $2$ and $n.$ We denote by
$E_i^*$ the event that $i$ is the dividing index of the string, and we sample strings with dividing index $i$ with probability $p(E_i^*)\propto(i-1)^k$ for $k$ a constant $-2<k<-3/2$.
With probability $p_+$ the value of the bit at the dividing index itself is $1$, and with probability $1-p_+$ it is $0$. There is one index, chosen uniformly at random from among indices less than the dividing index, such that the bit at that index has value $1$. All other bits in the string have value $0$.

Let $E_i$ be the event that the bit at index $i$ has value $1$. Then note that, given the description above, 
\begin{align}\label{eq:distribution_description}
p(E_i|E_j^*)=
\begin{cases}
1/(i-1) & \textrm{ if }i<j,\\
p_+ & \textrm{ if } j=i,\\
0 & \textrm{ if }i>j.
\end{cases}
\end{align}

Let $A_n$ be the normalizing factor of $p(E_i^*)$, so $p(E_i^*)=A_n(i-1)^k$

\begin{lemma} \label{lem:class_in_order}
For the distribution described around \cref
{eq:distribution_description}, to find all $1$-valued bits in the string, or to
find the first $1$-valued bit in the string, querying the bits of the string in order will result in an algorithm whose average query complexity is within one query of the optimal strategy.
\end{lemma}

\begin{proof}
We first show that as long as no $1$-valued bit has yet been found in the string $x$, 
the most likely place for a $1$-valued bit to be found is at the unqueried bits with the smallest or second smallest indices.

Let $S$ be the set of un-queried indices of $x$, and as just stated, we know
 that $x_i=0$ for all $i\notin
S$. Let $s_l$ be
the $l^\textrm{th}$ smallest element of $S$.

We will show that for $i,i'\in S$, if $s_1<i<i'$ then
\begin{equation}
p(E_i|S)-p(E_{i'}|S)>0,
\end{equation}
where $p(E_i|S)$ is the probability that $x_i=1$, given that $S$ are the unqueried indices (and by assumption all other indices are queried and 
have been found to have value $0$.) This implies that the most likely
place to find a bit with value $1$ among all unqueried indices except $s_1$ is $s_2$, regardless of which prior queries have been made.

Summing over all possible locations of dividing index, we have
\begin{align}\label{eq:cond_sum_prob}
p(E_i|S)&=\sum_{j=1}^np(E_j^*|S)p(E_i|E_j^*,S)\nonumber\\
&=\sum_{j\geq i}p(E_j^*|S)p(E_i|E_j^*,S).
\end{align}
where we have used \cref{eq:distribution_description} in the second line.
Note that $p(E_i|E_j^*,S)$ takes the same value for all $i\in S$ such that 
$i<j$, because the first $1$-valued bit is uniformly distributed over all indices less than the dividing index.
Thus when we analyze $p(E_i|S)-p(E_{i'}|S)$, we get a cancellation of terms
in the summation, giving us, for $i>s_1$,
\begin{align}\label{eq:prob_diff}
p(E_i|S)-p(E_{i'}|S)=p_+\left(p(E_i^*|S)-p(E_{i'}^*|S)\right)+\sum_{j:i'\geq j>i}p(E_j^*|S)p(E_i|E_j^*,S),
\end{align}
where we have replaced $p(E_i|E_i^*,S)$ and $p(E_{i'}|E_{i'}^*,S)$ with $p_+$ by \cref{eq:distribution_description}. Using Bayes' Theorem, we can write $p(E_j^*|S)$ as
\begin{equation}\label{eq:Bayes1}
p(E_j^*|S)=\frac{p(S|E_j^*)p(E_j^*)}{p(S)}.
\end{equation}

We first analyze $p(S|E_j^*)$. 
When the dividing index is $j$, since a priori, the first $1$-valued bit 
is equally distributed among the $j-1$ prior indices, the probability that
all bits with index less than $j$, except those in $S$, have value $0$ is
\begin{equation}\label{eq:cond_prob_part}
\frac{\left|\left\{l\in S: l<j\right\}\right|}{j-1},
\end{equation}
while the probability that $x_j$ has value $0$ is $(1-p_+)$. Since the likelihood of the dividing index bit being $0$ or $1$ is independent from the probability of prior bits being $0$ or $1$, we have
\begin{align}\label{eq:cond_prob_full}
p(S|E_j^*)=
\begin{cases}
\frac{\left|\left\{l\in S: l<i\right\}\right|}{j-1} &\textrm{ if } j\in S\\
\frac{(1-p_+)\left|\left\{l\in S: l<j\right\}\right|}{j-1} &\textrm{ if } j\notin S
\end{cases}
.
\end{align}
Thus using \cref{eq:cond_prob_full} and the power law distribution of $p(E_i^*)$, we can rewrite \cref{eq:Bayes1} as
\begin{align}\label{eq:Bayes_plug_in}
p(E_j^*|S)=
\begin{cases}
\frac{A_{n}\left|\left\{l\in S: l<j\right\}\right|(j-1)^{k-1}}{p(S)}&\textrm{ if } j\in S\\
\frac{A_{n}(1-p_+)\left|\left\{l\in S: l<j\right\}\right|(j-1)^{k-1}}{p(S)} &\textrm{ if } j\notin S
\end{cases}
.
\end{align}
By a similar argument as in \cref{eq:cond_prob_part}, we have that for $i<j,$
\begin{equation}\label{eq:cond_prob_2}
p(E_i|E_j^*,S)=\frac{1}{\left|\left\{l\in S: l<j\right\}\right|}.
\end{equation}

Plugging \cref{eq:Bayes_plug_in} and \cref{eq:cond_prob_2}, into \cref{eq:prob_diff} we have
\begin{align}\label{eq:analyzed_prob_diff}
p(E_i|S)-p(E_{i'}|S)=&
\frac{A_{n}}{p(S)}\left( p_+\left(\left|\left\{l\in S: l<i\right\}\right|(i-1)^{k-1}-\left|\left\{l\in S: l<i'\right\}\right|(i'-1)^{k-1}\right)\right.\nonumber\\
&\left.+\sum_{j\in S:i'\geq j>i}(j-1)^{k-1}+\sum_{j\notin S:i'\geq j>i}(1-p_+)(j-1)^{k-1}\right)
\nonumber\\
>&
\frac{A_{n}(i'-1)^{k-1}}{p(S)}
\left(p_+\left(\left|\left\{l\in S: l<i\right\}\right|-
\left|\left\{l\in S: l<i'\right\}\right|\right)+\left(\sum_{j\in S:i'\geq j>i}1\right)\right)
\nonumber\\
\geq & 0
\end{align}
where the first inequality comes from replacing all $(j-1)^{k-1}$ and 
$(i-1)^{k-1}$ terms with the smaller term $(i'-1)^{k-1}$. Therefore, the
probability of finding a $1$-valued bit is always higher at smaller
indices (up to the second to smallest unqueried index.)

Next we show
\begin{equation}
p(E_{s_1}|S)\geq p(E_{s_2}|S).
\end{equation}
Note that $p(E_{s_1}^*|S)=0$, since all bits with index less than $s_1$ have been queried and found to value $0$, so $s_1$
can not be the location of the dividing index because there must be a $1$-valued bit with index smaller than the dividing index. Thus we
modify \cref{eq:cond_sum_prob} to get
\begin{equation}
p(E_{s_1}|S)=\sum_{j\in S:j>i}p(E_j^*|S)p(E_i|E_j^*,S).
\end{equation}
Then using a similar analysis as in \cref{eq:prob_diff}, we have
\begin{align}\label{eq:prob_diff_first_two}
p(E_{s_1}|S)-p(E_{s_2}|S)=&
\frac{A_{n}}{p(S)}\left(-p_+\left|\left\{l\in S: l<s_2\right\}\right|(s_2-1)^{k-1}\right.\nonumber\\
&\left.+\sum_{j\in S:s_2\geq j>s_1}(j-1)^{k-1}+\sum_{j\notin S:s_2\geq j>s_1}(1-p_+)(j-1)^{k-1}\right)
\nonumber\\
\geq &
\frac{A_{n}}{p(S)}\left( -p_+\left|\left\{l\in S: l<s_2\right\}\right|(s_2-1)^{k-1}+(s_2-1)^{k-1}\right)
\nonumber\\
\geq& 0,
\end{align}
where in the second line, we have used that the second summation is non-negative, and the first summation only contains one term, $j=s_2.$ 
In the final line, we used that $\left|\left\{l\in S: l<s_2\right\}\right|=1$, since $s_1$ is the only element of $S$ with value less than $s_2$. Note
that there is only equality in the last line when $p_+=1.$

Thus combining \cref{eq:analyzed_prob_diff,eq:prob_diff_first_two}, we have that for $i>s_2$,
\begin{equation}
p(E_i|S)>p(E_{s_2}|S)\geq p(E_{s_1}|S),
\end{equation}
where again the final inequality can only be tight if $p_+=1$

Therefore the best strategy for a classical algorithm to find a $1$-valued bit is to always query the first or second unqueried index until a $1$ is found.
If the initial $1$ is found at the first unqueried index, then the algorithm
has found the first non-zero bit. If the algorithm needs to find an
additional $1$-valued bit, the probability of finding a $1$ at any later
index is given by the power law distribution, so the best strategy is to
query the remaining bits in order.
If the initial $1$ is found at the second as-yet-unqueried index, then with
one additional query, the algorithm can query the first as-yet-unqueried
index, to determine if the first $1$-valued bit is there. If it is, the
algorithm has found all $1$-valued inputs, including the first $1$-valued input. If it is not, then the algorithm has found the first $1$-valued input, and 
the probability of finding the next $1$ at any later index is given by the power
series distribution, so the best strategy is to query the remaining bits in
order.

Thus, even if we find a $1$ in the second
unqueried position, an asympototically optimal strategy is then to go back to
the first unqueried position and query it, since that only adds $1$ extra query to the complexity. This modified strategy with the addition of querying this extra bit, is equivalent in query complexity, if not worse than, querying the bits of the string in order. 
Thus without loss of generality, we can
assume that the optimal classical strategy to find either the first
$1$-valued bit, or both $1$-valued bits is to query the bits of the string in
order, as the algorithm does not get an asymptotic advantage (the difference is at most one query in the worst case).
\end{proof}

\begin{lemma}\label{lem:sum_calculations}
Given the distribution of bit strings as described around \cref{eq:distribution_description}, if $E_i^\dagger$ is the event that $i$ 
is the index of the first $1$-valued bit, then
\begin{align}
&\sum_{i=2}^np(E_i^*)\sqrt{i}=O(1)\label{eq:quantum_query_bound_star}\\
&\sum_{i=2}^np(E_i^*)i=\Omega(n^{k+2}\label{eq:classical_query_bound_star})\\
&\sum_{i=1}^{n-1}p(E_i^\dagger)\sqrt{i}=O(1) \label{eq:quantum_query_bound}\\
&\sum_{i=1}^{n-1}p(E_i^\dagger)i=\Omega(n^{k+2})\label{eq:classical_query_bound}.
\end{align}
\end{lemma}

\begin{proof}
We first bound $A_{n}$ where recall $A_{n}$ is the normalization factor such that $p(E_i^*)=A_{n}(i-1)^k$ and $\sum_{i=1}^{n}p(E_i^*)=1.$
From \cite{montanaro2010quantum}, we have that
\begin{align}
\frac{A_{n-1}}{k}\left(k+(n-1)^k-1\right)\geq A_{n}\geq& \frac{k+1}{(n-1)^{k+1}+k}\qquad .
\end{align}
Because $-2<k<-3/2,$ this tells us that $A_{n}=\Theta(1).$

We first prove \cref{eq:quantum_query_bound_star}. We have
\begin{align}
\sum_{i=2}^np(E_i^*)\sqrt{i}&=A_{n}\sum_{i=2}^n(i-1)^k\sqrt{i}\nonumber\\
&\leq A_{n}\left(\sqrt{2}+\int_{i=3}^{n+1}(i-2)^k\sqrt{i-1}di\right)\nonumber\\
&\leq A_{n}\left(\sqrt{2}+2\int_{i=3}^{n+1}(i-2)^{k+1/2}di\right)\nonumber\\
&\leq A_n\left(\sqrt{2}+\frac{4(n-1))^{k+3/2}}{3+2k}-\frac{4}{3+2k}\right)\nonumber\\
&=O(1).
\end{align}

For \cref{eq:classical_query_bound_star}, we have
\begin{align}
\sum_{i=2}^np(E_i^*)i&=A_{n}\sum_{i=2}^n(i-1)^ki\nonumber\\
&\geq A_{n}\int_{i=2}^{n+1}i^{k+1}di\nonumber\\
&= \frac{A_{n}}{k+2}\left((n+1)^{k+2}-2^{k+2}\right)\nonumber\\
&=\Omega\left(n^{k+2}\right)
\end{align}

To prove \cref{eq:classical_query_bound}, we first lower bound $p(E_i^\dagger)$
\begin{align}
p(E_i^\dagger)=&\sum_{j>i}p(E_j^*)p(E_i|E_j^*)\nonumber\\
&=A_{n}\sum_{j>i}(j-1)^{k-1}\nonumber\\
&>A_{n}\int_{i+1}^{n+1}(j-1)^{k-1}dj\nonumber\\
&=\frac{A_{n}}{k}\left(n^k-i^k\right),
\end{align}
where we've used \cref{eq:distribution_description} in the second line.
Thus 
\begin{align}
\sum_{i=1}^{n-1}p(E_i^\dagger)i&>A_{n}\sum_{i=1}^{n-1}\left(n^ki-i^{k+1}\right)\nonumber\\
&> \frac{A_{n}}{k}\left(\frac{n^{k+1}(n-1)}{2}-\int_1^{n}i^{k+1}di\right)\nonumber\\
&> \frac{A_{n}}{k}\left(\frac{n^{k+1}(n-1)}{2}-\frac{n^{k+2}-1}{k+2}\right)\nonumber\\
&> \frac{A_{n}}{2k(k+2)}\left(n^{2+k}k-(k+2)n^{k+1}+2\right)\nonumber\\
&=\Omega(n^{2+k}).
\end{align}

To prove \cref{eq:quantum_query_bound}, we upper bound $p(E_i^\dagger)$. We first do this for $i>1$:
\begin{align}
\label{eq:porpoise}
p(E_i^\dagger)=&\sum_{j>i}p(E_j^*)p(E_i|E_j^*)\nonumber\\
&=A_{n}\sum_{j>i}(j-1)^{k-1}\nonumber\\
&<A_{n}\left(\int_{i+1}^{n+1}(j-2)^{k-1}dj\right)\nonumber\\
&=\frac{A_{n}}{k}\left((n-1)^k-(i-1)^k\right)\nonumber\\
&\leq \frac{A_{n}}{k}\left(-(i-1)^k\right).
\end{align}
For $i=1$, we obtain
\begin{align}\label{eq:dolphin}
p(E_1^\dagger)=&\sum_{j>1}p(E_j^*)p(E_1|E_j^*)\nonumber\\
&=A_{n}\sum_{j>1}(j-1)^{k-1}\nonumber\\
&<A_{n}\left(1+\int_{3}^{n+1}(j-2)^{k-1}dj\right)\nonumber\\
&=\frac{A_{n}}{k}\left(k+(n-1)^k-1\right)\nonumber\\
&\leq \frac{A_{n}}{k}(k-1)
\end{align}

We combine \cref{eq:porpoise,eq:dolphin} to get
\begin{align}
\sum_{i=1}^np(E_i^\dagger)\sqrt{i}=&\frac{A_{n}}{k}\left(k-1-\sum_{i=2}^n(i-1)^k\sqrt{i}\right)\nonumber\\
\leq &\frac{A_{n}}{k}\left(k-1-2\sum_{i=2}^n(i-1)^k\sqrt{i-1}\right)\nonumber\\
\leq &\frac{A_{n}}{k}\left(k-1-2\sum_{i=1}^{n-1}i^{k+1/2}\right)\nonumber\\
\leq &\frac{A_{n}}{k}\left(k-1-2\int_{i=2}^n(i-1)^{k+1/2}di\right)\nonumber\\
\leq &\frac{A_{n}}{k}\left(k-1-4\left(\frac{(n-1)^{k+3/2}-1}{3+2k}\right)\right)\nonumber\\
=&O(1),
\end{align}
proving \cref{eq:quantum_query_bound}.
\end{proof}

\end{document}